%% file: jnl-2016-fdtd-dissipative.tex
\documentclass[journal]{IEEEtran}

\usepackage{amsmath}
\usepackage{tikz}
\usepackage{fix-cm}
\usetikzlibrary{arrows}

\usepackage{graphicx}

\usetikzlibrary{calc,positioning}
\usepackage[outline]{contour}
\contourlength{1.5pt}
\usepackage{amsthm}

\input{myDefinitions.tex}

\usepackage [english]{babel}
\usepackage [autostyle, english = american]{csquotes}
\MakeOuterQuote{"}
\usepackage{caption,subcaption}
\usepackage[caption = false] {subfig}
\usepackage[]{threeparttable}
\usepackage{cuted}

\usepackage{soul}

\newcommand{\half}{\frac{1}{2}}
\newcommand{\ihat}{\hat{\imath}}
\newcommand{\jhat}{\hat{\jmath}}

\newcommand{\E}{\mathcal{E}}

\newcommand{\dx}{\Delta x}
\newcommand{\dxhat}{\Delta \hat{x}}
\newcommand{\dyhat}{\Delta \hat{y}}
\newcommand{\dy}{\Delta y}
\newcommand{\dz}{\Delta z}
\newcommand{\dt}{\Delta t}

\newcommand{\DA}{\matr{D}_{A}}
\newcommand{\Dlx}{\matr{D}_{l_x}}

\newcommand{\Dly}{\matr{D}_{l_y}}

\newcommand{\Dlxp}{\matr{D}_{l'_x}}
\newcommand{\Dlyp}{\matr{D}_{l'_y}}

\newcommand{\eps}{\varepsilon}

\newcommand{\Dex}{\matr{D}_{\varepsilon_x}}
\newcommand{\Dexhat}{\matr{\hat{D}}_{\varepsilon_x}}
\newcommand{\Dey}{\matr{D}_{\varepsilon_y}}
\newcommand{\Dsx}{\matr{D}_{\sigma_x}}
\newcommand{\Dsxhat}{\matr{\hat{D}}_{\sigma_x}}
\newcommand{\Dsy}{\matr{D}_{\sigma_y}}

\newcommand{\Dm}{\matr{D}_{\mu}}

\newcommand{\Gx}{\matr{G}_x}
\newcommand{\Gy}{\matr{G}_y}

\newcommand{\T}{\vect{T}}
\newcommand{\ES}{E_N}
\newcommand{\EShat}{\vect{\hat{E}}_S}
\newcommand{\HS}{H_N}
\newcommand{\HShat}{\vect{\hat{H}}_S}
\newcommand{\Hzhat}{\vect{\hat{H}}_{\jhat+\half}}
\newcommand{\Hz}{H_{j-\half}}

\newcommand{\ESleft}{\frac{\dy}{2} \left( \frac{\eps_x}{\dt} + \frac{\sigma_x}{2} \right)}
\newcommand{\ESright}{\frac{\dy}{2} \left( \frac{\eps_x}{\dt} - \frac{\sigma_x}{2} \right)}
\newcommand{\EShatleft}{ \frac{\dy}{2r} \left( \frac{\Dexhat}{\dt} + \frac{\Dsxhat}{2} \right)}
\newcommand{\EShatright}{\frac{\dy}{2r} \left( \frac{\Dexhat}{\dt} - \frac{\Dsxhat}{2} \right)}

\newtheorem{theorem}{Theorem}
\newtheorem{definition}{Definition}

\definecolor{HV}{rgb}{0,0.3,0}
\newcommand{\hvcolorname}{green}
\definecolor{Hcol}{rgb}{0.7,0,0}
\definecolor{EB}{rgb}{0.7,0,0}
\definecolor{Ecol}{rgb}{0,0,1}

\definecolor{legendGreen}{rgb}{0,0.5,0}
\definecolor{legendRed}{rgb}{1,0,0}
\definecolor{legendBlue}{rgb}{0,0,1}
\definecolor{legendGrey}{rgb}{0.5,0.5,0.5}

\begin{document}
%
% paper title
% can use linebreaks \\ within to get better formatting as desired
% Do not put math or special symbols in the title.
\title{A Dissipative Systems Theory for FDTD with Application to Stability Analysis and Subgridding}

% author names and IEEE memberships
% note positions of commas and nonbreaking spaces ( ~ ) LaTeX will not break
% a structure at a ~ so this keeps an author's name from being broken across
% two lines.
% use \thanks{} to gain access to the first footnote area
% a separate \thanks must be used for each paragraph as LaTeX2e's \thanks
% was not built to handle multiple paragraphs
%

\author{Fadime~Bekmambetova,
        Xinyue~Zhang,
        and~Piero~Triverio,~\IEEEmembership{Senior Member,~IEEE}% <-this % stops a space
\thanks{Manuscript received ...; revised ...}%
\thanks{This work was supported in part by the Natural Sciences and Engineering Research Council of
Canada (Discovery grant program) and in part by the Canada Research Chairs program.}
\thanks{F.~Bekmambetova, X.~Zhang and P.~Triverio are with the Edward S. Rogers Sr. Department of Electrical and Computer Engineering, University of Toronto, Toronto, M5S 3G4 Canada (email: fadime.bekmambetova@mail.utoronto.ca, xinyuezhang.zhang@mail.utoronto.ca, piero.triverio@utoronto.ca).}
}

% The paper headers
\markboth{Journal of \LaTeX\ Class Files,~Vol.~11, No.~4, December~2012}%
{Shell \MakeLowercase{\textit{et al.}}: Bare Demo of IEEEtran.cls for Journals}
% The only time the second header will appear is for the odd numbered pages
% after the title page when using the twoside option.
% 
% *** Note that you probably will NOT want to include the author's ***
% *** name in the headers of peer review papers.                   ***
% You can use \ifCLASSOPTIONpeerreview for conditional compilation here if
% you desire.

% make the title area
\maketitle

\IEEEpeerreviewmaketitle
\begin{abstract}
This paper establishes a far-reaching connection between the Finite-Difference Time-Domain method (FDTD) and the theory of dissipative systems. The FDTD equations for a rectangular region are written as a dynamical system having the magnetic and electric fields on the boundary as inputs and outputs. Suitable expressions for the energy stored in the region and the energy absorbed from the boundaries are introduced, and used to show that the FDTD system is dissipative under a generalized Courant-Friedrichs-Lewy condition. Based on the concept of dissipation, a powerful theoretical framework to investigate the stability of FDTD methods is devised. The new method makes FDTD stability proofs simpler, more intuitive, and modular. Stability conditions can indeed be given on the individual components (e.g. boundary conditions, meshes, embedded models) instead of the whole coupled setup. As an example of application, we derive a new subgridding method with material traverse, arbitrary grid refinement, and guaranteed stability. The method is easy to implement  and has a straightforward stability proof. Numerical results confirm its stability, low reflections, and ability to handle material traverse.
\end{abstract}

\section{Introduction}
\label{sec:intro}

The Finite-Difference Time-Domain (FDTD) method is widely used to solve Maxwell's equations numerically in microwave and antenna engineering, photonics and physics~\cite{taflove2005computational, Gedney}. FDTD is versatile, easy to implement, and has a low computational cost per time step. Through explicit update equations, FDTD recursively computes the electric and magnetic field in the region of interest without requiring the solution of linear systems. This update process is stable if time step $\dt$ satisfies the Courant-Friedrichs-Lewy (CFL) stability limit~\cite{Gedney}
\begin{equation}
	\dt < \frac{1}{\sqrt{\frac{1}{\eps\mu}} \sqrt{\frac{1}{\dx^2} + \frac{1}{\dy^2}  + \frac{1}{\dz^2}}  }\,,
	\label{eq:CFL3d}
\end{equation}
where $\dx$, $\dy$, $\dz$ denote cell size, $\varepsilon$ denotes permittivity and $\mu$ denotes permeability.  While many breakthroughs have been achieved in FDTD since Yee's original algorithm~\cite{yee}, the efficiency of FDTD for multiscale problems remains an open problem. The simultaneous presence of large and small geometrical features can dramatically reduce FDTD efficiency because of two factors. First, the mesh has to be refined, at least locally, to properly resolve small features, which increases the number of unknowns and the cost per iteration. Second, because of~\eqref{eq:CFL3d}, a mesh refinement imposes a smaller time step, which further increases computational cost. For example, a 3X refinement of the entire FDTD mesh increases computational cost by 81 times. These issues are unfortunate since multiscale problems abound in practice.

Numerous solutions have been proposed to mitigate this issue, including sophisticated boundary conditions to mimic open spaces~\cite{Gedney}, local grid refinement~\cite{okoniewski1997three,thoma1996consistent,xiao2007three} (commonly known as \emph{subgridding}), thin wire models~\cite{taflove2005computational}, lumped elements~\cite{sui1992extending, edelvik2004general}, and hybridizations with model order reduction~\cite{denecker,jnl-2014-tmtt-fdtd,cnf-2015-nemo-fdtdmor}, finite elements~\cite{lee1997time}, integral equations~\cite{bretones1998hybrid}, ray tracing~\cite{wang2000hybrid} and implicit schemes like ADI-FDTD~\cite{adi-fdtd, adi-fdtd-2}. 
From a system theory viewpoint, many of these methods consist of the interconnection of \emph{subsystems}, such as models, algorithms, boundary conditions. For example, in~\cite{bretones1998hybrid}, an FDTD model, used to describe an inhomogeneous scatterer, is coupled to the time-domain method of moments, used to model a thin-wire antenna. In this way, one can leverage the respective strengths of different algorithms. However, while the stability properties of the individual algorithms may be well understood, ensuring the stability of their combination can be a formidable task. Ensuring stability is not trivial even in the relatively simple case of FDTD subgridding, where one just couples a coarse and a fine FDTD grid through an interpolation rule to relate fields at the grid transition. The lack of a systematic approach to ensure the stability of advanced FDTD methods is a major issue, that limits the development of new methods and their adoption by industry, where guaranteed stability is mandatory. This issue is the main motivation for this paper, that proposes a new theoretical framework to investigate and ensure the stability of both simple and advanced FDTD methods. 

Several techniques are available to analyze the stability of FDTD-like methods. Von Neumann analysis~\cite{taflove2005computational} is the simplest, but is only applicable to uniform meshes and homogeneous materials. The iteration method~\cite{Gedney} checks if the eigenvalues of the matrix that relates the current and next field solution are all below one in magnitude. In the energy method~\cite{edelvik2004general} instead, one writes an expression for the total energy stored in the simulation domain, and then checks if the algorithm satisfies a discrete equivalent of the principle of energy conservation. Since energy conservation prevents an unphysical growth of the solution, this implies stability~\cite{jnl-2007-tadvp-fundamentals}. The iteration and energy methods are general, but they can lead to long derivations, since they require the analysis of the whole coupled scheme, which may consist of many subsystems such as FDTD meshes, different boundary conditions, reduced order models, and so on. For example, in a subgridding scenario, one must derive the iteration matrix or energy function of the whole scheme, taking simultaneously into account coarse and fine meshes, interpolation rules and boundary conditions. This issue makes stability analysis quite involved. Moreover, it does not yield stability conditions on the individual subsystems. Consequently, if a subsystem is changed (for example, a different boundary condition is introduced), the iteration matrix or energy function of the whole problem must be derived again.

In this paper, we propose a new stability framework  for FDTD. The framework is based on the theory of dissipative systems~\cite{willems1972dissipative}, and generalizes the energy method. First, starting from FDTD update equations, we develop a self-contained mathematical model for a region with arbitrary permittivity, permeability and conductivity. The model is in the form of a discrete time dynamical system. The magnetic field tangential to the boundary is taken as input, while the electric field tangential to the boundary is taken as output. Through this form, we reveal that an FDTD model can be interpreted as a dynamical system which is dissipative when $\dt$ satisfies a generalized CFL condition. If this condition is not met, the system can generate energy on its own, leading to an unstable simulation. We thus establish a  connection between FDTD and the elegant theory of dissipative systems, which is a novel result. We believe that this connection will greatly benefit the FDTD community, since the theory of dissipative systems has been extremely successful in control theory for ensuring the stability of interconnected systems. This key result sets the basis for a powerful FDTD stability theory, where each part of a given FDTD setup (standard meshes, boundary conditions, interpolation schemes, ...) is interpreted as a subsystem, and is required to be dissipative. Since the connection of dissipative systems is dissipative~\cite{willems1972dissipative,jnl-2007-tadvp-fundamentals}, this will ensure the stability of the FDTD algorithm resulting from the connection of the subsystems. The proposed theory has numerous advantages. It simplifies stability proofs, since conditions can be imposed on \emph{each subsystem individually}, rather than on the whole coupled algorithm. Stability proofs are thus made modular and ``reusable'': once a given FDTD model (e.g., an advanced boundary condition) has been deemed to be dissipative, it can be combined to any other dissipative FDTD subsystem with guaranteed stability. The proposed approach also naturally provides the CFL stability limit of the resulting scheme, which will be the most restrictive CFL limit of the individual subsystems. Finally, the theory is intuitive, since it is based on the concept of energy, familiar to most scientists. The proposed theory is presented in 2D, for the sake of clarity. An extension to 3D is feasible and is currently under development.

As an example of application, the proposed theory is used to derive the subgridding algorithm which is stable by construction, and has several desirable features. Material traverse is supported for both dielectrics and highly conductive materials, which is a limitation of other stable subgridding methods~\cite{wang2010analysis}. The proposed method has low reflections, and avoids non-rectangular cells~\cite{xiao2007three}, finite element concepts~\cite{collino2006conservative} and Withney forms~\cite{chilton2007conservative,venkatarayalu2007stable}. Corners are natively supported without any special treatment~\cite{thoma1996consistent} nor L-shaped cells~\cite{xiao2007three}. Ultimately, the proposed algorithm just consists of a compact FDTD-like update equation for the edges between coarse and fine mesh, and is thus easy to implement. This update equation is provided explicitly for an arbitrary integer refinement ratio $r$, while several previous works~\cite{thoma1996consistent,xiao2007three} provide the update weights only for specific refinements (typically $r=2$ or $3$), leaving the derivation of other cases to the Reader.

The paper is organized as follows. In Sec.~\ref{sec:fdtdsystem}, we cast the FDTD update equations for a rectangular region into the form of a dynamical system with suitable inputs and outputs. In Sec.~\ref{sec:dissipative}, we show that FDTD equations can be interpreted as a dissipative system, and propose the new stability theory. The theory is applied to derive a stable subgridding algorithm in Sec.~\ref{sec:subgridding}, followed by numerical results in Sec.~\ref{sec:numerical_examples}.

\section{Discrete Time Dynamical Model for a 2D FDTD Region}
\label{sec:fdtdsystem}

The goal of this section is to cast the FDTD equations for a 2D region into the form of a discrete time dynamical model. The model shall be self-contained, involving only field samples from the nodes belonging to the region. This goal will be achieved by introducing suitable magnetic field samples at the boundaries. 

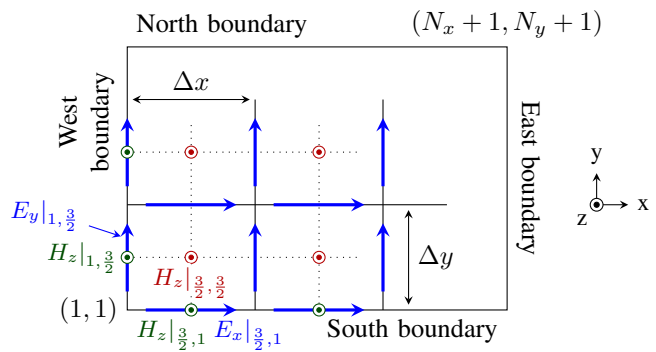
\begin{figure}[t]
	\centering
	\begin{tikzpicture}
	[sgrid/.style={dotted},
	arrE/.style={very thick,->,Ecol,>=stealth,shorten >=7pt, shorten <=7pt},
	labelH/.style={font=\fontsize{\sizefont}{\sizefont}\selectfont, Hcol},
	labelHV/.style={font=\fontsize{\sizefont}{\sizefont}\selectfont, HV},
	arr/.style={thick,->,>=stealth,shorten >=7pt, shorten <=7pt},
	labelE/.style={font=\fontsize{\sizefont}{\sizefont}\selectfont, Ecol},
	label/.style={font=\fontsize{\sizefont}{\sizefont}\selectfont}
	]
	\pgfmathsetmacro{\sizefont} {9};
	\pgfmathsetmacro{\W}{1.7};
	\pgfmathsetmacro{\H}{1.4};
	\pgfmathsetmacro{\extraW}{0.8};
	\pgfmathsetmacro{\extraH}{0.7};
	\pgfmathsetmacro{\Rbig}{0.08};
	\pgfmathsetmacro{\Rsmall}{0.03};
	
%	%draw axis
%	\draw [->,>=stealth] (0,0) -- (2.5*\W+\extraW,0);
%	\draw [->,>=stealth] (0,0) -- (2.5*\W+\extraW,0);
	
	%primary cell boundaries
	\draw[] (0,0) rectangle (2.5*\W+\extraW, 2*\H+\extraH);
	\draw (\W,0) -- (\W,2*\H);
	\draw (2*\W,0) -- (2*\W,2*\H);
	\draw (0,\H) -- (2.5*\W,\H);
	
	%secondary cell boundaries
	\draw [sgrid] (0, 0.5*\H) -- (1.8*\W,0.5*\H);
	\draw [sgrid] (0, 1.5*\H) -- (1.8*\W,1.5*\H);
	\draw [sgrid] (0.5*\W, 0) -- (0.5*\W,1.8*\H);
	\draw [sgrid] (1.5*\W, 0) -- (1.5*\W,1.8*\H);
	
	%E-fields
	\draw [arrE] (0,0) -- (1*\W, 0);
	\draw [arrE] (1*\W,0) -- (2*\W, 0);
	\draw [arrE] (0,\H) -- (1*\W, \H);
	\draw [arrE] (1*\W,\H) -- (2*\W, \H);
	\draw [arrE] (0,0) -- (0, \H);
	\draw [arrE] (\W,0) -- (\W, \H);
	\draw [arrE] (2*\W,0) -- (2*\W, \H);
	\draw [arrE] (0,\H) -- (0, 2*\H);
	\draw [arrE] (\W,\H) -- (\W, 2*\H);
	\draw [arrE] (2*\W,\H) -- (2*\W, 2*\H);
		
	%internal H-fields
	\draw [Hcol, fill=white](\W/2,\H/2) circle (\Rbig);
	\draw [Hcol, fill=white](1.5*\W,\H/2) circle (\Rbig);
	\draw [Hcol, fill=white](\W/2,1.5*\H) circle (\Rbig);
	\draw [Hcol, fill=white](1.5*\W,1.5*\H) circle (\Rbig);
	
	\draw [fill,Hcol] (\W/2,\H/2) circle (\Rsmall);
	\draw [fill,Hcol] (1.5*\W,\H/2) circle (\Rsmall);
	\draw [fill,Hcol] (\W/2,1.5*\H) circle (\Rsmall);
	\draw [fill,Hcol] (1.5*\W,1.5*\H) circle (\Rsmall);
	
	%hanging H-fields
	\draw [HV, fill=white](0.5*\W,0*\H) circle (\Rbig);
	\draw [HV, fill=white](1.5*\W,0*\H) circle (\Rbig);
	\draw [HV, fill=white](0*\W,0.5*\H) circle (\Rbig);
	\draw [HV, fill=white](0*\W,1.5*\H) circle (\Rbig);
	
	\draw [HV, fill](0.5*\W,0*\H) circle (\Rsmall);
	\draw [HV, fill](1.5*\W,0*\H) circle (\Rsmall);
	\draw [HV, fill](0*\W,0.5*\H) circle (\Rsmall);
	\draw [HV, fill](0*\W,1.5*\H) circle (\Rsmall);
	
	%label the internal H 
	\draw [labelH] (0.5*\W, 0.5*\H) node [below] {\contour{white}{$H_z|_{\frac{3}{2},\frac{3}{2}}$}};
	%label hanging H-fields
	\draw [labelHV](0.35*\W,0*\H) node [below] {$H_z|_{\frac{3}{2},1}$};
	\draw [labelHV](0*\W,0.5*\H) node [left] {$H_z|_{1,\frac{3}{2}}$};
	
	%label E-fields
	%\draw [->,>=stealth, Ecol] (0.85*\W,-0.25*\H) -- (0.75*\W,-0.05*\H);
	%\draw [labelE](1.1*\W,-0.25*\H) node [below] {$E_x|_{\frac{3}{2},1}$};
	\draw [labelE](0.95*\W,0*\H) node [below] {$E_x|_{\frac{3}{2},1}$};
	
	\draw [->,>=stealth, Ecol] (-0.3*\W,0.8*\H) -- (-0.05*\W,0.7*\H);
	\draw [labelE](-0.3*\W,0.9*\H) node [left] {$E_y|_{1,\frac{3}{2}}$};

	%label (1,1) and (Nx+1, Ny+1)
	\draw (0,0) node [left] {$(1,1)$};
	\draw (2.5*\W+\extraW, 2*\H+\extraH) node [above] {\contour{white}{$(N_x+1, N_y+1)$}};
	
	%label dx and dy
	\draw [<->, >=stealth, shorten >=2pt, shorten <=2pt] (0,2*\H) -- (\W,2*\H);
	\draw [<->, >=stealth, shorten >=2pt, shorten <=2pt] (2.2*\W,0) -- (2.2*\W,\H);
	
	\draw (\W/2,2*\H) node [above] {$\dx$};
	\draw (2.2*\W,0.5*\H) node [right] {$\dy$};

	%South boundary
	\draw (2.5*\W+\extraW, 0) node [below left] {South boundary};
	%North boundary
	\draw (0, 2*\H+\extraH) node [above right] {North boundary};
	%West boundary
	\draw (0, 2*\H+\extraH) node [above left, rotate=90] {\begin{tabular}{c} West\\boundary\end{tabular}};
	%East boundary
	\draw (2.5*\W+\extraW, \H+0.5*\extraH) node [above, rotate=270] {East boundary};

	% draw coordinate axes and label them
	\pgfmathsetmacro{\axisOffsetx}{0.7*\W};
	\pgfmathsetmacro{\axisOffsety}{\H};
	\draw[->,black,>=stealth] (2.5*\W+\extraW+\axisOffsetx, 0+\axisOffsety) -- (2.5*\W+\extraW+\axisOffsetx, 0.3*\H+\axisOffsety);
	\draw[->,black,>=stealth] (2.5*\W+\extraW+\axisOffsetx, 0+\axisOffsety) -- (2.5*\W+\extraW+\axisOffsetx + 0.3*\H, 0+\axisOffsety);
	\draw [fill=white] (2.5*\W+\extraW+\axisOffsetx, \axisOffsety) circle (0.05*\W);
	\draw [fill=black] (2.5*\W+\extraW+\axisOffsetx, \axisOffsety) circle (0.02*\W);
	\draw[label] (2.5*\W+\extraW+\axisOffsetx + 0.3*\H, 0+\axisOffsety) node [right] {x};
	\draw[label] (2.5*\W+\extraW+\axisOffsetx, 0+\axisOffsety+ 0.3*\H) node [above] {y};
	\draw[label] (2.5*\W+\extraW+\axisOffsetx, 0+\axisOffsety) node [below left] {z};
	
	\end{tikzpicture}
	\caption{Graphical representation of the 2D region considered in Sec.~\ref{sec:fdtdsystem}. The hanging variables introduced on the four boundaries are denoted in \hvcolorname.}
	\label{fig:cell}
\end{figure}

We consider the 2D rectangular region shown in Fig.~\ref{fig:cell}, operating in a TE mode with components $E_x$, $E_y$, and $H_z$. The region is discretized with a uniform rectangular grid with $N_x~\times~N_y$ cells, of width $\dx$ and height $\dy$. In addition to the field samples used by standard FDTD, we also sample the $H$ field on the four boundaries of the region. These additional samples will be referred to as \emph{hanging variables}~\cite{venkatarayalu2007stable}, and will allow us to:
\begin{enumerate}
\item develop a self-contained model for the region, which does not involve field samples beyond its boundaries;
\item derive an expression for the energy absorbed from each boundary;
\item connect the FDTD grid to other subsystems while maintaining stability.
\end{enumerate}

To keep the notation compact, we collect all $E_x$ and $E_y$ samples into column vectors $\vect{E}_x^n$ and $\vect{E}_y^n$, of size $N_{E_x}=N_x(N_y+1)$ and $N_{E_y}=(N_x+1) N_y$, respectively. The $H_z$ samples at the internal nodes are collected into column vector $\vect{H}_z^{n+\half}$ of size $N_{H_z} = N_x N_y$. The hanging variables on the South boundary of the region are collected into the $N_x \times 1$ vector
\begin{equation}
 \vect{H}_{S}^{n+\half} = 
 \begin{bmatrix}
 	H_z|_{1+\half,1}^{n+\half} & \hdots & H_z|_{N_x+\half,1}^{n+\half}
 \end{bmatrix}^T\,.
\end{equation} 
Similarly, the hanging variables on the North, East and West sides are cast into column vectors 
  $\vect{H}_{N}^{n+\half}$, $\vect{H}_{E}^{n+\half}$ and $\vect{H}_{W}^{n+\half}$,  respectively.

\subsection{State Equation for Each Node}

The dynamical model for the region consists of an update equation for each E and H sample, excluding the hanging variables. Those variables will indeed be eliminated once the region is connected to the surrounding subsystems or to some boundary conditions.

\subsubsection{Internal $H_z$ Nodes}
For these nodes, we use a standard FDTD update~\cite{Gedney}
\begin{multline}
	\dx \dy \frac{\mu}{\dt} H_z|_{i+\half, j+\half}^{n+\half} = 
			\dx \dy \frac{\mu}{\dt} H_z|_{i+\half, j+\half}^{n-\half} \\
			- \dx E_x|_{i+\half,j}^n + \dx E_x|_{i+\half,j+1}^n \\
			+ \dy E_y|_{i,j+\half}^n - \dy E_y|_{i+1,j+\half}^n 
	\label{eq:updateH}
\end{multline}
for $i=1,\dots,N_x$ and $j=1,\dots,N_y$. In~\eqref{eq:updateH}, $\mu$ denotes the average permittivity on the edge where $H_z|_{i+\half, j+\half}^{n+\half}$ is sampled. Subscripts are omitted from $\mu$ in order to simplify the notation. Equations~\eqref{eq:updateH} can be written in matrix form as
\begin{equation}
	\DA \frac{\Dm}{\dt} \vect{H}_z^{n+\half} = 
	\DA \frac{\Dm}{\dt} \vect{H}_z^{n-\half}
	+ \Gy \Dlx \vect{E}_x^n - \Gx \Dly \vect{E}_y^n\,,
	\label{eq:updateHmat}
\end{equation}
where $\DA$ is a diagonal matrix containing the area of the primary cells, $\Dm$ is a diagonal matrix containing the average permittivity on each edge of the secondary grid. Diagonal matrices
\begin{align}
	\Dlx = \dx \matr{I}_{N_{E_x}}  && 	\Dly = \dy \matr{I}_{N_{E_y}} 
	\label{eq:DlxDly}
\end{align} 
contain the length of the $x$- and $y$-directed edges of the primary grid, respectively. With $\matr{I}_m$, we denote the $m\times m$ identity matrix. Matrix $\Gx$ is the discrete derivative operator along $x$, which can be written as
\begin{equation}
	\matr{G}_x = \matr{I}_{N_y} \otimes \matr{W}_{N_x}\,,
	\label{eq:Gx}
\end{equation}
where $\otimes$ is the Kronecker's product~\cite{brewer} and $\matr{W}_{N_x}$ is the $N_x \times (N_x +1)$ matrix~\cite{denecker-statespace}
\begin{equation}
	\matr{W}_{n} = 
		\begin{bmatrix}
			-1 & +1 &    &   &  \\
			   & -1 & +1 &   &  \\
			   &    & \ddots & \ddots & \\
			   &    &        &   -1   & +1
		\end{bmatrix}
		\,.
		\label{eq:W}
\end{equation}
Similarly, 
\begin{equation}
	\Gy = \matr{W}_{N_y} \otimes \matr{I}_{N_x}\,.
	\label{eq:Gy}
\end{equation}

\subsubsection{$E_x$ Nodes}

For the $E_x$ nodes that fall strictly inside the region, we use a standard FDTD update~\cite{Gedney}
\begin{multline}
	\dx \dy \left( \frac{\eps_x}{\dt} + \frac{\sigma_x}{2} \right) E_x|_{i+\half,j}^{n+1} = 
		\dx \dy \left( \frac{\eps_x}{\dt} - \frac{\sigma_x}{2} \right) E_x|_{i+\half,j}^{n} \\
		+ \dx H_z|_{i+\half,j+\half}^{n+\half} - \dx H_z|_{i+\half,j-\half}^{n+\half} 
	\label{eq:updateEx}
\end{multline}
for $i = 1,\dots,N_x$ and $j = 2,\dots,N_y$. In~\eqref{eq:updateEx}, $\eps_x$ and $\sigma_x$ denote, respectively, the average permittivity and average conductivity on the corresponding edge. While the common factor $\dx$ could be eliminated from~\eqref{eq:updateEx}, it is kept as it will be useful later. 

For the $E_x$ nodes on the South boundary, the standard FDTD update equation involves a $H_z$ sample outside the region. In order to avoid this, we apply the finite difference approximation to the half step of the secondary grid between nodes $\left( i+\half,1+\half \right)$ and $\left( i+\half,1 \right)$, making use of the hanging variables. In this way, we obtain an update equation that involves only field samples from the considered region
\begin{multline}
	\dx \frac{\dy}{2} \left( \frac{\eps_x}{\dt} + \frac{\sigma_x}{2} \right) E_x|_{i+\half,1}^{n+1} = \\
		\dx \frac{\dy}{2} \left( \frac{\eps_x}{\dt} - \frac{\sigma_x}{2} \right) E_x|_{i+\half,1}^{n} \\
		+ \dx H_z|_{i+\half,1+\half}^{n+\half} - \dx H_z|_{i+\half,1}^{n+\half}\,,
	\label{eq:updateExS}
\end{multline}
for $i = 1,\dots,N_x$. In~\eqref{eq:updateExS}, $\eps_x$ and $\sigma_x$ denote the material properties of the half cell between nodes $\left( i+\half,1+\half \right)$ and $\left( i+\half,1 \right)$. The update equation for the $E_x$ nodes on the North boundary is obtained similarly, and reads
\begin{multline}
	\dx \frac{\dy}{2} \left( \frac{\eps_x}{\dt} + \frac{\sigma_x}{2} \right) E_x|_{i+\half,N_y+1}^{n+1} = \\
		\dx \frac{\dy}{2} \left( \frac{\eps_x}{\dt} - \frac{\sigma_x}{2} \right) E_x|_{i+\half,N_y+1}^{n} \\
		+ \dx H_z|_{i+\half,N_y+1}^{n+\half} - \dx H_z|_{i+\half,N_y+\half}^{n+\half}\,,
	\label{eq:updateExN}
\end{multline}
for $i = 1,\dots,N_x$.
Relations~\eqref{eq:updateEx}, \eqref{eq:updateExS} and~\eqref{eq:updateExN} can be compactly written as
\begin{multline}
	\Dlx \Dlyp \left( \frac{\Dex}{\dt} + \frac{\Dsx}{2} \right) \vect{E}_x^{n+1} = \\
	\Dlx \Dlyp \left( \frac{\Dex}{\dt} - \frac{\Dsx}{2} \right) \vect{E}_x^{n} 
	- \Dlx \Gy^T \vect{H}_z^{n+\half} \\
	+ 
	\begin{bmatrix}
		\Dlx \matr{B}_{S} & \Dlx \matr{B}_{N}
	\end{bmatrix}	 
	\begin{bmatrix}
		\vect{H}_{S}^{n+\half} \\
		\vect{H}_{N}^{n+\half}
	\end{bmatrix}	 \,,
	\label{eq:updateExmat}
\end{multline}
where
\begin{itemize}
	\item $\Dlyp$ is an $N_{E_x} \times N_{E_x}$ diagonal matrix containing the length of the $y$-directed edges of the secondary grid, including the half-edges of length $\dy/2$ that intersect the North and South boundaries;
	\item $\Dex$ and $\Dsx$ are diagonal matrices storing the permittivity and conductivity on each $x$-directed primary edge, respectively;
	\item $\matr{B}_S$ has all entries set to zero, except for a $-1$ at the intersection of each row associated with a South boundary edge and the column of the corresponding hanging variable in $\vect{H}_S^{n+\half}$;
	\item $\matr{B}_N$ is defined similarly to $\matr{B}_S$, but has a $+1$ on each entry associated with an edge of the North boundary.
\end{itemize}

\subsubsection{$E_y$ Nodes}

The update equations for the $E_y$ nodes are derived with a similar procedure. A standard FDTD update equation is written for the $E_y$ nodes that fall strictly inside the region. For the $E_y$ nodes on the West and East boundaries, instead, finite differences are applied on a half grid step, similarly to~\eqref{eq:updateExS} and~\eqref{eq:updateExN}. This process leads to 
\begin{multline}
	\Dly \Dlxp \left( \frac{\Dey}{\dt} + \frac{\Dsy}{2} \right) \vect{E}_y^{n+1} = \\
	\Dly \Dlxp \left( \frac{\Dey}{\dt} - \frac{\Dsy}{2} \right) \vect{E}_y^{n} 
	+ \Dly \Gx^T \vect{H}_z^{n+\half} \\
	+ 
	\begin{bmatrix}
		\Dly \matr{B}_{W} & \Dly \matr{B}_{E}
	\end{bmatrix}	 
	\begin{bmatrix}
		\vect{H}_{W}^{n+\half} \\
		\vect{H}_{E}^{n+\half}
	\end{bmatrix}	 \,,
	\label{eq:updateEymat}
\end{multline}
where
\begin{itemize}
	\item $\Dlxp$ is an $N_{E_y} \times N_{E_y}$ diagonal matrix containing the length of the $x$-directed edges of the secondary grid, including the half-edges of length $\dx/2$ that intersect the East and West boundaries;
	\item $\Dey$ and $\Dsy$ are diagonal permittivity and conductivity matrices associated with the $y$-directed edges of the primary grid, respectively;
	\item $\matr{B}_W$ has all entries set to zero, except for a $+1$ at the intersection of each row associated with a West boundary edge and the column of the corresponding hanging variable in $\vect{H}_W^{n+\half}$;
	\item $\matr{B}_E$ is analogous to $\matr{B}_W$, but has a $-1$ on each entry associated with an East boundary edge.
\end{itemize}

\subsection{Descriptor System Formulation}

The update equations derived in the previous sections form a complete dynamical model for the rectangular region
\begin{subequations}
\begin{eqnarray}
	(\matr{R} + \matr{F}) \vect{x}^{n+1} & = &  (\matr{R} - \matr{F}) \vect{x}^{n} + \matr{B} \vect{u}^{n+\half}\,, \label{eq:sys1a} \\
	\vect{y}^n & = & \matr{L}^T \vect{x}^n\,.
	\label{eq:sys1b}
\end{eqnarray}
\end{subequations}
The first equation~\eqref{eq:sys1a} is formed by update equations~\eqref{eq:updateHmat}, \eqref{eq:updateExmat} and~\eqref{eq:updateEymat}, and updates the state vector
\begin{equation}
	\vect{x}^{n} =
	\begin{bmatrix}
		\vect{E}_x^n \\
		\vect{E}_y^n \\
		\vect{H}_z^{n-\half}
	\end{bmatrix}\,,
	\label{eq:x}
\end{equation}
which consists of all electric and magnetic field samples in the region, excluding hanging variables. The input vector $\vect{u}^{n+\half}$ and output vector $\vect{y}^{n}$ are given by
\begin{equation}
	\vect{u}^{n+\half} = 
	\begin{bmatrix}
		\vect{H}_S^{n+\half} \\
		\vect{H}_N^{n+\half} \\
		\vect{H}_W^{n+\half} \\
		\vect{H}_E^{n+\half} \\						
	\end{bmatrix}
	\quad
		\vect{y}^{n} = 
	\begin{bmatrix}
		\vect{E}_S^{n} \\
		\vect{E}_N^{n} \\
		\vect{E}_W^{n} \\
		\vect{E}_E^{n} \\						
	\end{bmatrix}\,.
	\label{eq:uy}
\end{equation}
The input vector contains all hanging variables, i.e. all magnetic field samples on the region boundaries.  The output vector is made by the E samples at the same nodes, which are collected into vectors	$\vect{E}_S^{n}$, $\vect{E}_N^{n}$, $\vect{E}_W^{n}$ and $\vect{E}_E^{n}$. Output equation~\eqref{eq:sys1b} 
extracts these values from the state vector $\vect{x}^n$. The coefficients matrices in~\eqref{eq:sys1a} and~\eqref{eq:sys1b} read
\begin{equation}
	\matr{R} = 
	\begin{bmatrix}
		\Dlx \Dlyp \frac{\Dex}{\dt}		& \matr{0}							& \half \Dlx \Gy^T 			\\
		\matr{0}							& \Dly \Dlxp \frac{\Dey}{\dt}		& -\half \Dly \Gx^T 			\\	
		\half \Gy \Dlx						& -\half \Gx \Dly					& \DA \frac{\Dm}{\dt}
	\end{bmatrix}\,,
	\label{eq:R}
\end{equation} 
\begin{equation}
	\matr{F} = 
	\begin{bmatrix}
		\Dlx \Dlyp \frac{\Dsx}{2}		& \matr{0}							& \half \Dlx \Gy^T 			\\
		\matr{0}							& \Dly \Dlxp \frac{\Dsy}{2}		& -\half \Dly \Gx^T 			\\	
		-\half \Gy \Dlx						& \half \Gx \Dly					& \matr{0}
	\end{bmatrix}\,,
	\label{eq:F}
\end{equation} 
\begin{equation}
	\matr{B} = 
	\begin{bmatrix}
		\Dlx \matr{B}_S		& \Dlx \matr{B}_N		& \matr{0}				& \matr{0}			\\
		\matr{0}			& \matr{0}				& \Dly \matr{B}_W		& \Dly \matr{B}_E	\\
		\matr{0}			& \matr{0}				& \matr{0}				& \matr{0}			\\
	\end{bmatrix}\,,
	\label{eq:B}
\end{equation} 
\begin{equation}
	\matr{L} = 
	\begin{bmatrix}
		-\matr{B}_S			& \matr{B}_N		& \matr{0}				& \matr{0}			\\
		\matr{0}			& \matr{0}			& \matr{B}_W			& - \matr{B}_E		\\
		\matr{0}			& \matr{0}			& \matr{0}				& \matr{0}			\\
	\end{bmatrix}\,.
	\label{eq:L}
\end{equation} 

Equations~\eqref{eq:sys1a}-\eqref{eq:sys1b} define an FDTD-like model for a 2D lossy, source-free region, with possibly non-uniform permittivity and permeability. From a control perspective, this model is a discrete time descriptor system~\cite{dai1989singular}, also known as generalized state-space. From an electrical standpoint, the model is an impedance-type description of the region, because it gives the electric field tangential to the boundary in terms of the tangential magnetic field. In~\cite{denecker-statespace}, a special case of the proposed representation was derived for a lossy region enclosed by PEC boundaries. The proposed developments generalize~\cite{denecker-statespace} in two directions:
\begin{enumerate}
	\item the assumption of a PEC termination is removed. Suitable input-output variables are introduced in such a way that model~\eqref{eq:sys1a}-\eqref{eq:sys1b} can be connected to other subsystems, such as different FDTD grids, boundary conditions, or reduced models;
	\item the proposed formulation will allow us to prove, in the next section, that FDTD equations can be interpreted as a dissipative dynamical system.
\end{enumerate}

\section{FDTD as a Dissipative Discrete Time System}
\label{sec:dissipative}
In this section, we define suitable expressions for the energy stored in the 2D region and the energy absorbed from its four boundaries. These expressions are then used to derive the conditions under which system~\eqref{eq:sys1a}-\eqref{eq:sys1b} is dissipative.

\subsection{Dissipation Inequality}

A discrete time system like~\eqref{eq:sys1a}-\eqref{eq:sys1b} is dissipative when it satisfies the following condition~\cite{byrnes1994losslessness}.
\begin{definition}
\label{def:dissipative}
Dynamic system~\eqref{eq:sys1a}-\eqref{eq:sys1b} is said to be dissipative with supply rate $s(\vect{y}^n, \vect{u}^{n+\half})$ if there exists a nonnegative function $\E(\vect{x}^n)$ with $\E(0) = 0$, called \emph{storage function}, such that 
\begin{equation}
	\E(\vect{x}^{n+1}) - \E(\vect{x}^n) \le s(\vect{y}^n, \vect{u}^{n+\half})
	\label{eq:dissipation}
\end{equation}
for all $\vect{u}^{n+\half}$ and all $n$.
\end{definition}

The storage function $\E(\vect{x}^n)$ can be interpreted as the energy stored in the system at time $n$, while the supply rate $s(\vect{y}^n, \vect{u}^{n+\half})$ is the energy absorbed by the system from its boundaries between time $n$ and $n+1$. Clearly, only dissipative systems can satisfy the \emph{dissipation inequality}~\eqref{eq:dissipation}. Indeed, their stored energy can increase between time $n$ and $n+1$ by \textit{at most} the energy absorbed from the outside world. If this limit is violated, the system is considered active, since it can generate energy on its own. For most practical systems, inequality~\eqref{eq:dissipation} is satisfied strictly because of the presence of losses.

\subsection{Storage Function and Supply Rate}

For system~\eqref{eq:sys1a}-\eqref{eq:sys1b}, we choose as candidate storage function
\begin{equation}
	\E(\vect{x}^n) = \frac{\dt}{2} \left (\vect{x}^n \right) ^T \matr{R} \vect{x}^n\,,
	\label{eq:storage1}
\end{equation}
and as supply rate
\begin{equation}
	s(\vect{y}^n, \vect{u}^{n+\half}) = \dt \frac{\left( \vect{y}^n + \vect{y}^{n+1} \right)^T}{2}  \matr{L}^T \matr{B} \vect{u}^{n+\half}\,.
	\label{eq:supply1}
\end{equation}
Before showing that these functions satisfy~\eqref{eq:dissipation}, we investigate their physical meaning. Substituting~\eqref{eq:R} into~\eqref{eq:storage1}, we get
\begin{multline} 
	\E(\vect{x}^n) = \half \left( \vect{E}_x^n \right)^T \Dlx \Dlyp \Dex \vect{E}_x^n  \\
		+  \half \left( \vect{E}_y^n \right)^T \Dly \Dlxp \Dey \vect{E}_y^n \\
		+ \frac{\dt}{2} \left[ \Gy \Dlx  \vect{E}_x^n  - \Gx \Dly \vect{E}_y^n +  \DA \frac{\Dm}{\dt} \vect{H}_z^{n-\half} \right] \,.
		\label{eq:storage2}
\end{multline}
Since the term between square brackets is the right hand side of~\eqref{eq:updateHmat}, we can reduce~\eqref{eq:storage2} to
\begin{multline} 
	\E(\vect{x}^n) = \half \left( \vect{E}_x^n \right)^T \Dlx \Dlyp \Dex \vect{E}_x^n  \\
		+  \half \left( \vect{E}_y^n \right)^T \Dly \Dlxp \Dey \vect{E}_y^n \\
		+   \half \left(  \vect{H}_z^{n-\half} \right)^T \DA \Dm \vect{H}_z^{n+\half}\,.
		\label{eq:storage3}
\end{multline}
 
\begin{equation}
\int_A \left[ \half  \eps E_x^2(t) + \half  \eps E_y^2(t) + \half  \mu H_z^2(t) \right]dA  \,,
\end{equation}
the continuous expression for the electromagnetic energy per unit height of a 2D TE mode in a region of area $A$. The first term of~\eqref{eq:storage3} is the energy stored in the $x$ component of the electric field in the region. Indeed, the diagonal matrix $\Dlx \Dlyp$ is the area of the half cells around the $E_x$ edges. Similarly, the second and third term in~\eqref{eq:storage3} represent the energy stored in the $E_y$ and $H_z$ components, respectively.

We now discuss the physical meaning of supply rate~\eqref{eq:supply1}. Direct inspection reveals that
\begin{equation}
	\matr{L}^T \matr{B} = 
	\begin{bmatrix}
		- \dx \matr{I}_{N_x}		& \matr{0} 					& \matr{0}  				& \matr{0} \\
		 \matr{0} 					&  + \dx \matr{I}_{N_x}		& \matr{0}  				& \matr{0} \\
		 \matr{0}  					& \matr{0} 					& + \dy \matr{I}_{N_y}		& \matr{0} 			\\
		 \matr{0} 					& \matr{0}  				& \matr{0} 					& - \dy \matr{I}_{N_y}	\\

	\end{bmatrix} \,.
	\label{eq:LB}
\end{equation}
Substituting~\eqref{eq:LB} into~\eqref{eq:supply1}, we obtain
\begin{multline}
	s(\vect{y}^n, \vect{u}^{n+\half}) = 
	- \dt \dx \frac{(\vect{E}_S^n + \vect{E}_S^{n+1})^T}{2} \vect{H}_S^{n+\half} \\
	+ \dt \dx \frac{(\vect{E}_N^n + \vect{E}_N^{n+1})^T}{2} \vect{H}_N^{n+\half} 
	+ \dt \dy \frac{(\vect{E}_W^n + \vect{E}_W^{n+1})^T}{2} \vect{H}_W^{n+\half} \\
	- \dt \dy \frac{(\vect{E}_E^n + \vect{E}_E^{n+1})^T}{2} \vect{H}_E^{n+\half} \,,
	\label{eq:supply2}
\end{multline}
and see that the supply rate is the sum of the energy absorbed by the region from each boundary between time $n$ and $n+1$.
Signs in~\eqref{eq:supply2} are consistent with the direction of the Poynting vector on each boundary.

\subsection{Dissipativity Conditions}

Using the proposed storage function~\eqref{eq:storage1} and supply rate~\eqref{eq:supply1}, we can derive simple dissipativity conditions on the coefficients matrices $\matr{R}$, $\matr{F}$, $\matr{B}$ and $\matr{L}$ in~\eqref{eq:sys1a}-\eqref{eq:sys1b}.
\begin{theorem}
\label{thm:dissipative1}
	If
	\begin{subequations}
	\begin{eqnarray}
		 & \matr{R} = \matr{R}^T >  0\,,  \label{eq:cond1} \\
		 & \matr{F} + \matr{F}^T \ge 0\,, \label{eq:cond2}  \\
		 & \matr{B} =  \matr{L} \matr{L}^T \matr{B}\,,  \label{eq:cond3} 
	\end{eqnarray}
	\end{subequations}
	then system~\eqref{eq:sys1a}-\eqref{eq:sys1b} is dissipative according to Definition~\ref{def:dissipative}, with~\eqref{eq:storage1} as storage function and~\eqref{eq:supply1} as supply rate.
\end{theorem}
\begin{proof} 
Condition~\eqref{eq:cond1} makes storage function~\eqref{eq:storage1}  nonnegative for all $\vect{x}^n$, as required by Definition~\ref{def:dissipative}. Next, we show that if~\eqref{eq:cond2} and~\eqref{eq:cond3} hold, the dissipation inequality~\eqref{eq:dissipation} will hold as well. Substituting~\eqref{eq:storage1} and~\eqref{eq:supply1} into~\eqref{eq:dissipation}, we obtain
\begin{equation}
	\left(\vect{x}^{n+1} \right) ^T \matr{R} \vect{x}^{n+1} - \left(\vect{x}^n \right) ^T \matr{R} \vect{x}^n - \left( \vect{y}^n + \vect{y}^{n+1} \right)^T \matr{L}^T \matr{B} \vect{u}^{n+\half} \le 0 \nonumber\,.
\end{equation}
This inequality can be rewritten as
\begin{multline}
	\left(\vect{x}^{n+1} \right) ^T \matr{R} \left( \vect{x}^{n+1} - \vect{x}^n \right)
	+ \left( \vect{x}^{n+1} - \vect{x}^n \right)^T \matr{R} \vect{x}^n \\
	- \left( \vect{y}^n + \vect{y}^{n+1} \right)^T \matr{L}^T \matr{B} \vect{u}^{n+\half} \le 0 \label{eq:dissipproof1}
\end{multline}
From~\eqref{eq:sys1a}, we have that
\begin{equation}
	\matr{R} \left( \vect{x}^{n+1} - \vect{x}^n \right) = - \matr{F} \left( \vect{x}^{n+1} + \vect{x}^n \right) - \matr{B} \vect{u}^{n+\half}
	\label{eq:RF}
\end{equation}
Using~\eqref{eq:RF}, the terms $\matr{R} \left( \vect{x}^{n+1} - \vect{x}^n \right)$ and $\left( \vect{x}^{n+1} - \vect{x}^n \right)^T \matr{R}$ in~\eqref{eq:dissipproof1} can be expressed in terms of $\matr{F}$
\begin{multline}
	- \left(\vect{x}^{n+1} \right) ^T \matr{F} \left( \vect{x}^{n+1} + \vect{x}^n \right)
	- \left( \vect{x}^{n+1} + \vect{x}^n \right)^T \matr{F}^T \vect{x}^n \\
	+ \left( \vect{x}^{n+1} \right)^T \matr{B} \vect{u}^{n+\half}  
	+  \left( \vect{u}^{n+\half} \right)^T \matr{B}^T \vect{x}^n \\ 
	- \left( \vect{x}^n \right)^T  \matr{L} \matr{L}^T \matr{B} \vect{u}^{n+\half} 
	- \left( \vect{x}^{n+1} \right)^T \vect{L} \vect{L}^T \vect{B} \vect{u}^{n+\half} \le 0 \label{eq:dissipproof2}
\end{multline}
Under~\eqref{eq:cond3}, we can finally rewrite~\eqref{eq:dissipproof2} as
\begin{equation}
	\left(\vect{x}^{n+1} +\vect{x}^n \right) ^T 
	\left( \matr{F} +\matr{F}^T \right)
	\left(\vect{x}^{n+1} +\vect{x}^n \right) \ge 0
\end{equation}
which is clearly satisfied if~\eqref{eq:cond2} holds.
\end{proof}

We now investigate the physical meaning of the three dissipativity conditions, starting from~\eqref{eq:cond3}, which can be directly verified by combining~\eqref{eq:LB} and~\eqref{eq:L}. This condition holds because the region inputs and outputs in~\eqref{eq:uy} are sampled at the same nodes, and is reminiscent of a similar relation which holds for linear circuits under the impedance representation~\cite{passivemacromodeling}. 
Condition~\eqref{eq:cond2} is related to losses, and reads
\begin{equation}
	\matr{F} + \matr{F}^T = 
	\begin{bmatrix}
		\Dlx \Dlyp \Dsx		& \matr{0}			& \matr{0}	\\
		\matr{0}			& \Dly \Dlxp \Dsy	& \matr{0}	\\	
		\matr{0}			& \matr{0}			& \matr{0}
	\end{bmatrix}
	\ge 0	\,.
	\label{eq:cond2b}
\end{equation} 
Since $\Dlx$, $\Dly$, $\Dlxp$ and $\Dlyp$ are diagonal with positive elements only, \eqref{eq:cond2b} boils down to
\begin{eqnarray}
	\Dsx & \ge & 0\,,	\\
	\Dsy & \ge & 0\,.	
\end{eqnarray}
As one may expect, the average conductivity on each primary edge must be non-negative in a dissipative system.
Condition~\eqref{eq:cond1} is the most interesting, and is a generalized CFL criterion. From~\eqref{eq:R}, we see that $\matr{R} = \matr{R}^T$ by construction. With the Schur's complement~\cite{Boy94}, we can transform~\eqref{eq:cond1} into
\begin{align}
& 	\Dlx \Dlyp \Dex > 0	\,,	
\label{eq:cond1a} \\
& 	\Dly \Dlxp \Dey > 0	\,,	
\label{eq:cond1b} \\
& 	\matr{S} \matr{S}^T < \frac{4}{\dt^2} \matr{I}_{N_{H_z}}\,,	
\label{eq:cond1c}
\end{align}
where
\begin{equation}
	\matr{S} = \DA^{-\half} \Dm^{-\half} 
		\begin{bmatrix}
			\Gy \Dlx^{\half} \Dlyp^{-\half} \Dex^{-\half} 	& 	-\Gx \Dly^{\half} \Dlxp^{-\half} \Dey^{-\half} 
		\end{bmatrix}\,.
		\label{eq:S}
\end{equation}
Inequalities~\eqref{eq:cond1a} and~\eqref{eq:cond1b} are clearly satisfied by construction. If we denote with $s_k$ the singular values of~\eqref{eq:S}, inequality~\eqref{eq:cond1c} will hold if
\begin{equation}
	\dt < \frac{2}{s_k} \quad \forall k\,.
	\label{eq:dtlimit}
\end{equation}
This condition sets an upper bound on the FDTD time step, and is analogous to the generalized CFL constraint derived in~\cite{denecker-statespace} for an FDTD grid terminated on PEC boundaries. From~\eqref{eq:S}, we indeed see that the time step upper limit depends on grid size, permittivity and permeability. Using a time step which violates~\eqref{eq:dtlimit} has two consequences. First, the discretization of Maxwell's equations leads to an active discrete time model, even if the real physical system has positive losses everywhere, and should therefore be dissipative. The numerical model is thus inconsistent with the actual physics. Second, being able to generate energy on its own, model~\eqref{eq:sys1a}-\eqref{eq:sys1b} can lead to divergent transient simulations. Even when such model is connected to other dissipative subsystems, its ability to generate energy on its own can destabilize the whole simulation. The proposed theory provides a new, deeper understanding on the root causes of FDTD instability, and is able to pinpoint \emph{which} part of an FDTD setup is responsible for it. It is hoped that this new explanation will facilitate the development of powerful FDTD schemes with guaranteed stability. 

\subsection{Relation to CFL Stability Limit}

The relation between~\eqref{eq:dtlimit} and the CFL limit can be further understood if we apply~\eqref{eq:dtlimit} to a single cell. The purpose of~\eqref{eq:dtlimit} is to make storage function~\eqref{eq:storage1} non-negative. A sufficient condition for this is to require the energy stored in \emph{each} primary cell to be positive. This condition can be derived by applying~\eqref{eq:cond1c} to a  single primary cell. 
We consider the primary cell extending from node $(i,j)$ to node $(i+1, j+1)$.
The coefficient matrices for this special case can be obtained from the formulas in Sec.~\ref{sec:fdtdsystem} by setting $N_x = N_y = 1$
\begin{align}
	\Dlx & = \dx \matr{I}_2  					& 		\Dly & = \dy \matr{I}_2 \label{eq:cell1} \\
	\Dlxp & = \frac{\dx}{2} \matr{I}_2  		& 		\Dlyp & = \frac{\dy}{2} \matr{I}_2 \label{eq:cell2} \\
	\Gx &= \Gy = \begin{bmatrix}
				-1 & 1
			\end{bmatrix}
												&  \DA & = \dx \dy \label{eq:cell3} 
\end{align}
The permeability matrix $\Dm = \mu|_{i+\half, j+\half}$ is the average permeability on the secondary edge, while the permittivity matrices
\begin{align}
	\Dex & = 	\begin{bmatrix}
					\!\eps_x|_{i+\half, j}\!\!\!	& 0 						\\
					0						& \!\!\! \eps_x|_{i+\half, j+1}\!
				\end{bmatrix}
	&
	\Dey & = 	\begin{bmatrix}
					\! \eps_y|_{i, j+\half} \!\!\!	& 0 						\\
					0						& \!\!\! \eps_y|_{i+1, j+\half} \!
				\end{bmatrix}				
				\label{eq:cell4} 			
\end{align} 
contain the average permittivity on the four primary edges.
Substituting~\eqref{eq:cell1}-\eqref{eq:cell4} into~\eqref{eq:cond1c}, we obtain
\begin{multline}
	\dt < \Bigg [ \frac{1}{2 \dx^2 \mu|_{i+\half, j+\half}} \left( \frac{1}{\eps_y|_{i, j+\half}} + \frac{1}{\eps_y|_{i+1, j+\half} } \right) \\
		+  \frac{1}{2 \dy^2 \mu|_{i+\half, j+\half}} \left( \frac{1}{\eps_x|_{i+\half, j}} + \frac{1}{\eps_x|_{i+\half, j+1} } \right)  \Bigg ]^{-\half}\,,
		\label{eq:CFLnonuniform}
\end{multline}
which is a generalized CFL condition. For a uniform medium with permittivity $\eps$ and permeability $\mu$, \eqref{eq:CFLnonuniform} reduces to
\begin{equation}
	\dt < \frac{1}{\sqrt{\frac{1}{\eps\mu}} \sqrt{\frac{1}{\dx^2} + \frac{1}{\dy^2}}  }\,,
	\label{eq:CFL2d}
\end{equation}
the CFL limit of 2D~FDTD. This derivation confirms that~\eqref{eq:cond1} is a generalized CFL condition, here reinterpreted in the context of dissipation.

\subsection{Application to Stability Analysis}
The proposed dissipation theory can be effectively used to investigate and enforce the stability of FDTD algorithms. Most FDTD setups consist of an interconnection of FDTD subsystems. In the simplest scenario, a uniform FDTD grid is connected to some boundary conditions. In most advanced scenarios, one may want to couple a main FDTD grid to refined grids, reduced models, lumped elements, or models from other numerical techniques, such as finite elements, integral equations or ray tracing. 

Ensuring stability of these hybrid schemes can be very challenging, since stability is a property of the overall scheme, rather than of its individual subsystems. By invoking the concept of dissipation, we can instead achieve stability in an easy and modular way. Each part is seen as an FDTD subsystem and required to satisfy dissipativity conditions~\eqref{eq:cond1}-\eqref{eq:cond3}. Since the connection of dissipative systems is dissipative by construction~\cite{jnl-2007-tadvp-fundamentals}, the overall method will be guaranteed to be stable. The proposed theoretical framework generalizes the so-called energy method~\cite{edelvik2004general}, and has numerous advantages over the state of the art:
\begin{enumerate}
	\item non-uniform problems can be handled, unlike in the von Neumann analysis~\cite{taflove2005computational};
	\item stability conditions can be given on each subsystem separately, unlike in the iteration method~\cite{Gedney}, which requires the analysis of the iteration matrix of the whole scheme. This makes stability analysis modular and thus simpler;
	\item once some given FDTD models have been proven dissipative, they can be arbitrarily interconnected without having to carry out further stability proofs. With the iteration method, when a single part of a coupled scheme changes, the whole proof must be revised;
	\item the CFL limit of the resulting scheme can be easily determined by applying~\eqref{eq:cond1a}, \eqref{eq:cond1b} or~\eqref{eq:cond1c} to each subsystem and taking the most restrictive CFL limit;
	\item the stability framework is intuitive, since it is based on the fundamental physical concept of energy dissipation.
\end{enumerate}

\section{Application: Stable FDTD Subgridding}
\label{sec:subgridding}
We demonstrate the proposed theory by deriving a subgridding algorithm which is stable by construction, easy to implement and supports an arbitrary grid refinement ratio. The goal is to derive stable update equations for a setup where one or more fine grids are embedded in a main coarse grid. Without loss of generality, we consider the case where a coarse grid with cell size $\dx \times \dy$ hosts a single fine grid with cell size $\dx/r \times \dy/r$, where $r$ is an arbitrary integer. The algorithm will ultimately consist of conventional FDTD equations to update the fields that fall strictly inside the two grids, and a special update equation to update the fields on the edges at the grid transition. To derive the method, it is sufficient to consider the interface between a single coarse cell and the corresponding fine cells, as shown in Fig.~\ref{fig:connectionA}. In the figure a virtual gap has been opened between the two grids for clarity. Without loss of generality, we consider a refinement in the positive $x$ direction. The other three cases can be derived in the same way. The coarse cell under consideration is centered at node $(i+\half,j-\half)$ and the corresponding fine cells at nodes $(\ihat+\half, \jhat + \half)$, ... , $(\ihat+r-\half, \jhat+\half)$, where coordinates $(\ihat, \jhat)$ correspond to the same physical location as $(i, j)$ in the coarse grid. Superscript "$\string^$" denotes variables related to the fine grid. 

A subgridding algorithm can be interpreted as a the result of the connection of the three subsystems, as shown in Fig.~\ref{fig:connectionB}. Two of those subsystems correspond to the coarse and fine grids to be coupled. The third subsystem represents the interpolation rule which is used to relate the fields on the boundaries of the two grids, that are sampled with different resolution.

\begin{figure}[t]
	
%	\begin{minipage}{\columnwidth}
		\centering
	%	\vspace*{\fill}
	\begin{tikzpicture}
	[
	arr/.style={very thick,->,Ecol,>=stealth,shorten >=10pt, shorten <=10pt},
	arrb/.style={very thick,->,Ecol,>=stealth,shorten >=10pt, shorten <=10pt},
	labelE/.style={font=\fontsize{\sizefont}{\sizefont}\selectfont, Ecol},
	labelH/.style={font=\fontsize{\sizefont}{\sizefont}\selectfont, Hcol},
	labelHV/.style={font=\fontsize{\sizefont}{\sizefont}\selectfont, HV},
	label/.style={font=\fontsize{\sizefont}{\sizefont}\selectfont}
	]

	\pgfmathsetmacro{\sizefont} {9}
	
	\pgfmathsetmacro{\W}{2.7}
	\pgfmathsetmacro{\H}{1.6}
	\pgfmathsetmacro{\dH}{0}%{0.05}
	\pgfmathsetmacro{\dHb}{0}%{0.07}
	\pgfmathsetmacro{\Rbig}{0.08}
	\pgfmathsetmacro{\Rsmall}{0.03}
	\pgfmathsetmacro{\extraW}{\H/4}
	\pgfmathsetmacro{\extraH}{\H/4}
	\pgfmathsetmacro{\axisOffsetx}{\W/4}
	\pgfmathsetmacro{\axisOffsety}{-0.75*\H}
	
	\pgfmathsetmacro{\gap}{0.2};

	\draw (-\extraW,-\gap) -- (2*\W+\extraW,-\gap);
	\draw (-\extraW,0) -- (2*\W+\extraW,0);
	\draw (-\extraW,\H) -- (2*\W+\extraW,\H);
	
	\draw (0, \H+\extraH) -- (0,0);
	\draw (0,-\gap) -- (0, -1*\H-\extraH-\gap);
	\draw (2*\W, \H+\extraH) -- (2*\W, 0);
	\draw (2*\W, -\gap) -- (2*\W, -1*\H-\extraH-\gap);
	
	\draw (\W, \H+\extraH) -- (\W, 0);
	\draw (2*\W, \H+\extraH) -- (2*\W, 0);

	% regular H-fields
	\draw [Hcol, fill=white] (1*\W, -1*\H-\gap) circle (\Rbig);
	\draw [Hcol, fill=white] (0.5*\W, 0.5*\H) circle (\Rbig);
	\draw [Hcol, fill=white] (1.5*\W, 0.5*\H) circle (\Rbig);
	
	\draw [fill, Hcol] (1*\W, -1*\H-\gap) circle (\Rsmall);
	\draw [fill, Hcol] (0.5*\W, 0.5*\H) circle (\Rsmall);
	\draw [fill, Hcol] (1.5*\W, 0.5*\H) circle (\Rsmall);

	%label H-fields =================
	%interhal h_hat
	\draw[Hcol] (\W/2,\H/2) node [above] {$\hat{H}_z|_{\ihat+\half, \jhat+\half}$};
	\draw[Hcol] (1.5*\W,\H/2) node [above] {$\hat{H}_z|_{\ihat+\frac{3}{2}, \jhat+\half}$};

	%hanging h_hat
	\draw[labelHV] (0.25*\W,\dH) node [above] {$\hat{H}_z|_{\ihat+\half, \jhat}$};
	\draw[labelHV] (1.25*\W,\dH) node [above] {$\hat{H}_z|_{\ihat+\frac{3}{2}, \jhat}$};
	
	%internal H
	\draw[labelH] (1*\W,-1*\H-\gap) node [above] {$H_z|_{i+\half, j-\half}$};
	%hanging H
	\draw[labelHV] (1*\W,-\dHb-\gap) node [below] {$H_z|_{i+\half, j}$};
	
	% label E-fields ================
	% e_hat
	\draw[labelE] (0.75*\W,\dH) node [above] {$\hat{E}_x|_{\ihat+\half, \jhat}$};
	\draw[labelE] (1.75*\W,\dH) node [above] {$\hat{E}_x|_{\ihat+\frac{3}{2}, \jhat}$};
	% E
	\draw[labelE] (1.65*\W,-\dHb-\gap) node [below] {$E_x|_{i+\half, j}$};
	
	%horizontal arrows
	\draw[arrb] (0,-\dHb-\gap) -- (2*\W,-\dHb-\gap);
	\draw[arr] (0,\dH) -- (\W,\dH);
	\draw[arr] (\W,\dH) -- (2*\W,\dH);
	
	% hanging vars
	\draw [HV, fill=white] (0.5*\W, \dH) circle (\Rbig);
	\draw [HV, fill=white] (1.5*\W, \dH) circle (\Rbig);
	\draw [HV, fill=white] (1*\W, -\dHb-\gap) circle (\Rbig);
	
	\draw [fill, HV] (0.5*\W, \dH) circle (\Rsmall);
	\draw [fill, HV] (1.5*\W, \dH) circle (\Rsmall);
	\draw [fill, HV] (1*\W, -\dHb-\gap) circle (\Rsmall);

	% label the (i,j) and (ihat, jhat)
	\draw[label] (0,0) node [above left] {$(\ihat,\jhat)$};
	\draw[label] (0,-\gap) node [below left] {$(i,j)$};
	
	% draw coordinate axes and label them
	\draw[->,black,>=stealth] (2*\W+\extraH+\axisOffsetx, 0+\axisOffsety) -- (2*\W+\extraH+\axisOffsetx, 0.2*\H+\axisOffsety);
	\draw[->,black,>=stealth] (2*\W+\extraH+\axisOffsetx, 0+\axisOffsety) -- (2*\W+\extraH+\axisOffsetx + 0.2*\H, 0+\axisOffsety);
	\draw [fill=white] (2*\W+\extraH+\axisOffsetx, \axisOffsety) circle (0.025*\W);
	\draw [fill=black] (2*\W+\extraH+\axisOffsetx, \axisOffsety) circle (0.01*\W);
	\draw[label] (2*\W+\extraH+\axisOffsetx + 0.2*\H, 0+\axisOffsety) node [right] {x};
	\draw[label] (2*\W+\extraH+\axisOffsetx, 0+\axisOffsety+ 0.2*\H) node [above] {y};
	\draw[label] (2*\W+\extraH+\axisOffsetx, 0+\axisOffsety) node [below left] {z};
	
	% label the dx/r and dy/r
	\draw[<->,black,>=stealth, shorten >=3pt, shorten <=3pt] (\W,\H+\extraH/3) -- (2*\W, \H+\extraH/3);	
	\draw[label] (1.5*\W, \H+\extraH/3) node [above] {$\frac{\dx}{r}$};
	
	\draw[<->,black,>=stealth, shorten >=3pt, shorten <=3pt] (2*\W+\extraW/3,0) -- (2*\W+\extraW/3, \H);	
	\draw[label] (2*\W+\extraW/3, \H/2) node [right] {$\frac{\dy}{r}$};	
	
	% label dx
	%\draw[<->,black,>=stealth, shorten >=3pt, shorten <=3pt] (0,-\H-\extraH/1) -- (2*\W, -\H-\extraH/1);	
	%\draw[label] (1.5*\W, -\H-\extraH/1) node [above] {$\dx$};
	
	\end{tikzpicture}
	\caption{Subgridding scenario considered in Sec.~\ref{sec:subgridding} for the case of $r=2$. For clarity, a virtual gap has been inserted between the two grids. This virtual gap is closed when the two grids are connected.}
	\label{fig:connectionA}
	
	\begin{tikzpicture}
	[sgrid/.style={dotted},
	arrE/.style={thick,->,Ecol,>=stealth,shorten >=7pt, shorten <=7pt},
	labelH/.style={font=\fontsize{\sizefont}{\sizefont}\selectfont, Hcol},
	labelHV/.style={font=\fontsize{\sizefont}{\sizefont}\selectfont, HV},
	labelE/.style={font=\fontsize{\sizefont}{\sizefont}\selectfont, Ecol},
	label/.style={font=\fontsize{\sizefont}{\sizefont}\selectfont}
	]
	\pgfmathsetmacro{\sizefont} {9};
	\pgfmathsetmacro{\H}{1};
	\pgfmathsetmacro{\Hbox}{0.7};
	\pgfmathsetmacro{\W}{6};
	\pgfmathsetmacro{\dw}{0.1};

	%coarse grid
	\draw (0,\H+\Hbox+\H) -- (\W,\H+\Hbox+\H);
	\draw[label] (\W/2,\H+\Hbox+\H) node [above] {Fine mesh};
	
	%fine grid
	\draw (0,0) -- (\W,0);
	\draw[label] (\W/2,0) [below] node {Coarse mesh};
	
	%interpolation rule
	\draw (0,\H) rectangle (\W, \H+\Hbox);
	\draw (\W/2, \H + \Hbox/2) node {Interpolation rule};
	
	% draw the arrows
	\draw[->,>=stealth, HV] (\W/2-\dw,0) -- (\W/2-\dw, \H);
	\draw[->,>=stealth, Ecol] (\W/2+\dw,\H) -- (\W/2+\dw, 0);
	
	\draw[->,>=stealth, HV] (\W/4-\dw,\H+\Hbox+\H) -- (\W/4-\dw, \H+\Hbox);
	\draw[->,>=stealth, Ecol] (\W/4+\dw,\H+\Hbox) -- (\W/4+\dw, \H+\Hbox+\H);
	
	\draw[->,>=stealth, HV] (3*\W/4-\dw,\H+\Hbox+\H) -- (3*\W/4-\dw, \H+\Hbox);
	\draw[->,>=stealth, Ecol] (3*\W/4+\dw,\H+\Hbox) -- (3*\W/4+\dw, \H+\Hbox+\H);
	
	% label the arrows
	\draw[labelHV] (\W/2-\dw, \H/2) node [left] {$H_z|_{i+\half, j}$};
	\draw[labelE] (\W/2+\dw,\H/2) node [right] {$E_x|_{i+\half, j}$};
	
	\draw[labelHV] (\W/4-\dw,\H+\Hbox+\H/2) node [left] {$\hat{H}_z|_{\ihat+\half, \jhat}$};
	\draw[labelE] (\W/4+\dw,\H+\Hbox+\H/2) node [right] {$\hat{E}_x|_{\ihat+\frac{3}{2}, \jhat}$};
	
	\draw[labelHV] (3*\W/4-\dw,\H+\Hbox+\H/2) node [left] {$\hat{H}_z|_{\ihat+\frac{3}{2}, \jhat}$};
	\draw[labelE] (3*\W/4+\dw,\H+\Hbox+\H/2) node [right] {$\hat{E}_x|_{\ihat+\half, \jhat}$};	

	\end{tikzpicture}
	\caption{Interpretation of the subgridding method as the connection of three dynamical systems, representing the coarse grid, the fine grid, and the interpolation rule.}
	\label{fig:connectionB}
\end{figure}

\subsection{State Equations for the Coarse and Fine Cells on the Boundary}

For the coarse and fine grids, we adopt the formulation of Sec.~\ref{sec:fdtdsystem}, introducing hanging variables on the two boundaries to be connected. The purpose of the hanging variables is to facilitate the coupling of the two meshes and the proof of stability. These extra variables will be eliminated when deriving the update equation for the fields at the interface. On the North boundary of the coarse cell, we introduce the hanging variable
%define hanging variables
\begin{equation}
\HS^{n+\half} = H_z|_{i+\half,j}^{n+\half}\,.
\end{equation} 
Similarly, on the South fine cell boundaries we introduce the hanging variables
\begin{equation}
\vect{\hat{H}}_{S}^{n+\half} = 
\begin{bmatrix}
\hat{H}_z|_{\ihat+\half,\jhat}^{n+\half} & \hdots & \hat{H}_z|_{\ihat+r-\half, \jhat}^{n+\half}
\end{bmatrix}^T\,,
\end{equation}
as shown in Fig.~\ref{fig:connectionA}. 

From (\ref{eq:updateExN}), we obtain the following state equation for the coarse E-field sample at the interface and the coarse hanging variable
\begin{multline}
\ESleft \ES^{n+1} = \\
\ESright \ES^{n} 
+ \HS^{n+\half}
- \Hz^{n+\half} 
\,,
\label{eq:updateExboundary}
\end{multline}
where
\begin{align}
\Hz^{n+\half} & = H_z|_{i+\half,j-\half}^{n+\half}\,,	& 	\ES^{n} & = E_x|_{i+\half,j}^{n}\,, %\\
%\Dlx & = \dx  												& 	\Dly & = \dy	\\
%\eps_x & = \eps_x|_{i+\half,j} 												& 	\sigma_x & = \sigma_x|_{i+\half,j}
\end{align}
and $\eps_x$ and $\sigma_x$ are the permittivity and conductivity on the primary edge of the coarse cell below the interface. Similarly, the state equation for the $r$ fine cells can be written as
\begin{multline}
\EShatleft \EShat^{n+1} = \\
\EShatright \EShat^{n} 
+ \Hzhat^{n+\half} - \vect{\hat{H}}_{S}^{n+\half}\,,
\label{eq:updateExhatboundary}
\end{multline}
where
%define internal fields
\begin{equation}
\Hzhat^{n+\half} = 
\begin{bmatrix}
\hat{H}_z|_{\ihat+\half,\jhat+\half}^{n+\half} & \hdots & \hat{H}_z|_{\ihat + r - \half,\jhat+\half}^{n+\half}
\end{bmatrix}^T\,,
\end{equation} 
\begin{equation}
\vect{\hat{E}}_{S}^{n} = 
\begin{bmatrix}
\hat{E}_x|_{\ihat+\half,\jhat}^{n} & \hdots & \hat{E}_x|_{\ihat+r-\half, \jhat}^{n}
\end{bmatrix}^T\,,
\end{equation} 
and $\Dexhat$ and $\Dsxhat$ are $r \times r$ diagonal matrices containing the values of permittivity and conductivity above the interface for the South edges where the fine $\hat{E}_x$ fields are sampled.

\subsection{Interpolation Rule}
The interpolation rule~\cite{Reciprocity} relates the field sampled on the coarse and fine grids. From the boundary condition for tangential electric fields, we have that
\begin{equation}
\EShat^n = \ES^n\T \qquad \forall n\,,
\label{eq:equalE}
\end{equation}
where $\T$ is an $r~\times~1$ matrix of ones. Condition~\eqref{eq:equalE} forces the coarsely- and finely-sampled $E$ fields to be equal at all times.

On the magnetic fields at the boundary, we impose a constraint reciprocal to~\eqref{eq:equalE}
\begin{equation}
\HS^{n+\half} = \frac{\T^T \HShat^{n+\half}}{r} \qquad \forall n\,.% = \frac{(\vect{\hat{H}}_S)^T\T}{r}
\label{eq:averH}
\end{equation} 
We will see in Sec.~\ref{sec:connectionStability} that the reciprocity between the $E$ and $H$ interpolation rules is required to ensure stability.

\subsection{Explicit Update Equation for the Interface}
Interpolation conditions~\eqref{eq:equalE} and~\eqref{eq:averH} can now be used to combine (\ref{eq:updateExboundary}) and (\ref{eq:updateExhatboundary}) in order to derive an explicit update equation for the fields at the coarse-fine interface, and eliminate hanging variables.

Substituting~\eqref{eq:equalE} into~\eqref{eq:updateExhatboundary}, and multiplying the obtained equation by $\T^T/r$ on the left yields
\begin{multline}
 \frac{\dy}{2r}  \left( \frac{\T^T \Dexhat \T}{r \dt} + \frac{\T^T \Dsxhat \T}{2r} \right)  \ES^{n+1}  = \\
 \frac{\dy}{2r}  \left( \frac{\T^T \Dexhat \T}{r \dt} - \frac{\T^T \Dsxhat \T}{2r} \right) \ES^{n} \\
+  \frac{\T^T \Hzhat^{n+\half}}{r}  - \frac{\T^T \vect{\hat{H}}_{S}^{n+\half}}{r} \,.
\label{eq:updateExinterm}
\end{multline}
For the simplicity of notation we define symbols for the average permittivity and conductivity of the $r$ South boundary fine cells
\begin{equation}
\hat{\eps}_x = \frac{\T^T \Dexhat \T}{r}\,,
\quad
\hat{\sigma}_x = \frac{\T^T \Dsxhat \T}{r}\,.
\end{equation}
Equation~\eqref{eq:updateExinterm} can now be added to (\ref{eq:updateExboundary}), yielding
\begin{multline}
\frac{\dy}{2}\left(\frac{\eps_x + \frac{\hat{\eps}_x}{r}}{\dt} + \frac{\sigma_x + \frac{\hat{\sigma}_x}{r}}{2}\right) \ES^{n+1} =\\
\frac{\dy}{2}\left(\frac{\eps_x + \frac{\hat{\eps}_x}{r}}{\dt} - \frac{\sigma_x + \frac{\hat{\sigma}_x}{r}}{2}\right)\ES^{n}\\
+ \HS^{n+\half} 
- \Hz^{n+\half} 
+ \frac{\T^T \Hzhat^{n+\half}}{r} 
- \frac{\T^T \HShat^{n+\half}}{r}
\,.
\label{eq:beforealmostUpdate}
\end{multline}
With the interpolation rule~\eqref{eq:averH}, we cancel the hanging variables to obtain
\begin{multline}
\frac{\dy}{2}\left(\frac{\eps_x + \frac{\hat{\eps}_x}{r}}{\dt} + \frac{\sigma_x + \frac{\hat{\sigma}_x}{r}}{2}\right) \ES^{n+1} = \\
\frac{\dy}{2}\left(\frac{\eps_x + \frac{\hat{\eps}_x}{r}}{\dt} - \frac{\sigma_x + \frac{\hat{\sigma}_x}{r}}{2}\right) \ES^{n}
-\Hz^{n+\half} 
+ \frac{\T^T \Hzhat^{n+\half}}{r} 
\,.
\label{eq:almostUpdate}
\end{multline}
Rearranging (\ref{eq:almostUpdate}), we get the following \emph{explicit} update equation for $\ES$ in terms of the neighboring magnetic fields %Am I messing it up?
\begin{multline}
\ES^{n+1} \!\! =  \!\!
			\left( 
				\frac{\eps_x + \frac{\hat{\eps}_x}{r}}{\dt} + \frac{\sigma_x + \frac{\hat{\sigma}_x}{r}}{2} 
			\right)^{\!\!-1}\!\!\!\!
			\left( 
				\frac{\eps_x + \frac{\hat{\eps}_x}{r}}{\dt} - \frac{\sigma_x + \frac{\hat{\sigma}_x}{r}}{2} 				\right) \!\!
			\ES^{n}
			 \\
			+ \frac{2}{\dy} \left( 
				\frac{\eps_x + \frac{\hat{\eps}_x}{r}}{\dt} + \frac{\sigma_x + \frac{\hat{\sigma}_x}{r}}{2} 
			\right)^{-1}
			\left( \frac{\T^T \Hzhat^{n+\half}}{r} - \Hz^{n+\half} \right)
\,.
\label{eq:connUpdate}
\end{multline}
The fine interface electric fields are then updated using~(\ref{eq:equalE}). It should be noted that, when $r=1$, equation~(\ref{eq:connUpdate}) reduces to the standard FDTD update equation. 

The overall subgridding algorithm can be summarized as follows:
\begin{enumerate}
	\item Calculate the magnetic field samples everywhere at time $n+\half$ using conventional FDTD update equations.
	\item Use standard FDTD update equations to compute the E fields at time $n+1$ on the edges that are strictly inside the coarse and fine grids.
	\item Compute $\ES^{n+1}$, the coarsely-sampled electric field at the interface, using~\eqref{eq:connUpdate}. 
	\item Update the finely-sampled $\hat{E}_x$ fields at the interface using~\eqref{eq:equalE}. 
\end{enumerate}
The computational overhead of this scheme is minimal, since the coefficients in~\eqref{eq:connUpdate} can be pre-computed before the update iterations. 

The proposed method is thus simple to implement, since it consists of conventional FDTD update equations inside the two meshes and a modified update equation for the edges at the interface. In comparison to previous subgridding methods, we avoid non-rectangular cells~\cite{xiao2007three}, finite element concepts~\cite{collino2006conservative} and Withney forms~\cite{chilton2007conservative,venkatarayalu2007stable}. The proposed update equation can be also used at corners with no modifications, unlike in previous works that require special treatment~\cite{thoma1996consistent} or L-shaped cells~\cite{xiao2007three}. Finally, we remark that in Sec.~\ref{sec:fdtdsystem}, the FDTD update equations have been given in matrix form in order to reveal the dissipative nature of FDTD systems. This form, however, does not have to be used in the practical implementation, which can use conventional for loops or, in languages like MATLAB, vectorized operations.

\subsection{Proof of Stability}
\label{sec:connectionStability}
The proposed dissipation theory makes it straightforward to prove that the subgridding algorithm is stable under the CFL limit of the fine grid. For stability, all three subsystems in Fig.~\ref{fig:connectionB} need to be dissipative, which requires one to use the more restrictive fine grid time step. With the time step chosen correctly, in order to guarantee the overall stability we need to only ensure dissipativity of the interpolation rule. 

Analogously to (\ref{eq:supply1}) the supply rate for the interpolation subsystem is defined as
\begin{multline}
s(\vect{y}^n, \vect{u}^{n+\half}) = 
- \dt \dx \frac{\ES^n + \ES^{n+1}}{2} \HS^{n+\half} \\
+ \dt \frac{\dx}{r} \frac{(\EShat^n + \EShat^{n+1})^T}{2} \HShat^{n+\half} \,.
\label{eq:supplyconnection}
\end{multline}
Substituting (\ref{eq:equalE}) and (\ref{eq:averH}) into (\ref{eq:supplyconnection}), we have
\begin{multline}
s(\vect{y}^n, \vect{u}^{n+\half}) = \\
\dt \dx \frac{(\ES^n + \ES^{n+1})^T}{2} 
\left(\frac{\T^T \HShat^{n+\half}}{r} -\HS^{n+\half}\right) = 0
\,.
\label{eq:supplyconnection2}
\end{multline}
Therefore, the proposed interpolation rule is a lossless system that does not dissipate nor absorb any energy. Physically, this result makes sense, since the connection system corresponds to an infinitely thin region where no energy dissipation can take place. In conclusion, since the proposed subgridding method can be seen as the connection of three dissipative systems, it is overall dissipative, and thus stable.

\section{Numerical Examples}
\label{sec:numerical_examples}
The following sections provide the results of FDTD simulations that were done to verify the proposed theory. The subgridding algorithm was implemented in Matlab and tests were performed in order to check its stability, ability to handle material traverse, its accuracy and speedup capability.

\subsection{Stability Verification}
\label{sec:ne_stability}
Stability was verified by simulating an empty cavity with perfect electric conductor (PEC) walls with a centrally placed subgridding region for 10$^6$ time steps. The layout of the simulation is shown in Fig.~\ref{fig:layout_stability}. The cavity was excited using a modulated Gaussian magnetic current source with central frequency of 3.75~GHz and half-width at half-maximum of 0.74~GHz. Magnetic field was recorded at a probe placed inside the cavity. The time step was set 1\% below the CFL limit of the fine grid.

The resulting waveform in Fig.~\ref{fig:result_stability} shows that no instability occurred after 10$^6$ time steps. Stable behavior after such a large number of time steps verifies the correctness of the proposed stability enforcement technique, especially since no lossy materials were present to dissipate any spurious energy artificially created by the algorithm.
\begin{figure}[t]
	\centering
	\begin{tikzpicture}
	[arr/.style={->,>=stealth,shorten >=1pt},
	label/.style={font=\fontsize{\sizefont}{\sizefont}\selectfont}
	]
	
\pgfmathtruncatemacro{\GR}{4};
\pgfmathtruncatemacro{\dxCmm}{1};
\pgfmathtruncatemacro{\dyCmm}{2};
\pgfmathsetmacro{\lenx}{4.8};
\pgfmathsetmacro{\leny}{3.2};
\pgfmathsetmacro{\subgrAx}{0.8};
\pgfmathsetmacro{\subgrAy}{0.8};
\pgfmathsetmacro{\subgrBx}{4};
\pgfmathsetmacro{\subgrBy}{2.4};
\pgfmathsetmacro{\curx}{0.6};
\pgfmathsetmacro{\cury}{0.4};
\pgfmathsetmacro{\probex}{4.2};
\pgfmathsetmacro{\probey}{2.8};
\pgfmathtruncatemacro{\truelenxmm}{60};
\pgfmathtruncatemacro{\truelenymm}{40};
\pgfmathtruncatemacro{\subgrlenxmm}{40};
\pgfmathtruncatemacro{\subgrlenymm}{20};
	
	\pgfmathsetmacro{\Rbig}{0.035}

	\pgfmathsetmacro{\sizefont} {9}
	%border + PEC label	
	\draw [very thick] (0, 0) rectangle (\lenx,\leny);
	\draw [->, >=stealth, shorten >=2pt, shorten <=7pt] (\lenx/2-0.5, \leny+0.4) -- (\lenx/2, \leny);
	\draw (\lenx/2-0.5, \leny+0.4) node {PEC};
	%Source
	\draw [red, fill = red] (\curx,\cury) circle (\Rbig);
	\draw[font=\fontsize{\sizefont}{\sizefont}\selectfont] (\curx,\cury) node [below] {Source};	
	%Probe
	\draw [blue, fill = blue] (\probex,\probey) circle (\Rbig);
	\draw[label] (\probex,\probey) node [left] {Probe};	
	
	%subgrids
	\draw [dashed, black, thick] (\subgrAx, \subgrAy) rectangle(\subgrBx, \subgrBy);

	%label the meshes
	\draw[label] (\lenx/2,\leny/2) node {\begin{tabular}{c} $\dxhat$ = $\dx/$\GR \\ $\dyhat$ = $\dy/$\GR \end{tabular}};
	\draw[label] (0,\leny) node [below right] {\begin{tabular}{c} $\dx$ = \dxCmm\ mm \\ $\dy$ = \dyCmm\ mm \end{tabular}};
	
	%axis
	\pgfmathsetmacro{\axisoffs}{\lenx/6}
	\draw[->,black,>=stealth] (\lenx + \axisoffs,\leny/2) -- (\lenx + \lenx/10+\axisoffs,\leny/2);
	\draw[->,black,>=stealth] (\lenx + \axisoffs,\leny/2) -- (\lenx + \axisoffs,\leny/2 + \lenx/10);
	\draw [fill=white] (\lenx + \axisoffs,\leny/2) circle (0.02*\lenx);
	\draw [fill=black] (\lenx + \axisoffs,\leny/2) circle (0.0075*\lenx);
	\draw[label] (\lenx + \axisoffs+\lenx/10,\leny/2) node [right] {x};
	\draw[label] (\lenx + \axisoffs,\leny/2 + \lenx/10) node [above] {y};
	\draw[label] (\lenx + \axisoffs,\leny/2) node [below left] {z};
	
	%label dimensions
	\draw[<->,black,>=stealth] (0,-0.15) -- (\lenx,-0.15);
	\draw (\lenx/2,-0.15) node [below] {\truelenxmm\ mm};
	\draw[<->,black,>=stealth] (-0.15,0) -- (-0.15,\leny);
	\draw (-0.15,\leny/2) node [above, rotate = 90] {\truelenymm\ mm};

	\draw[<->,black,>=stealth] (\subgrAx,\subgrAy-0.1) -- (\subgrBx,\subgrAy-0.1);
	\draw (\lenx/2,\subgrAy-0.1) node [below] {\subgrlenxmm\ mm};
	
	\draw[<->,black,>=stealth] (\subgrAx-0.1,\subgrAy) -- (\subgrAx-0.1,\subgrBy);
	\draw (\subgrAx-0.1,\leny/2) node [above, rotate=90] {\subgrlenymm\ mm};
	
	\end{tikzpicture}
	\caption{Layout of the PEC cavity considered in Sec.~\ref{sec:ne_stability}.}
	\label{fig:layout_stability}
\end{figure}
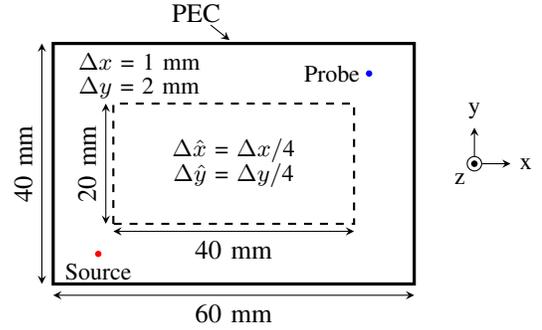

\begin{figure}[t]
	\centering
	\includegraphics[width=0.95\columnwidth]{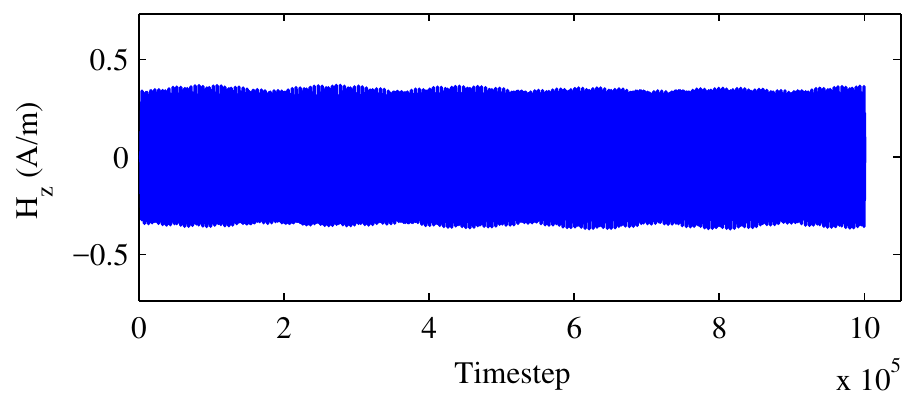}
	\caption{Magnetic field at the probe for the the empty cavity with subgridding of Sec.~\ref{sec:ne_stability}, computed for 10$^6$ time steps.}
	\label{fig:result_stability}
\end{figure}

\subsection{Material Traverse}
\label{sec:ne_mt}
The ability of the proposed method to produce meaningful results when objects traverse the subgridding interface was tested using the setup in Fig.~\ref{fig:layout_mt}. A 16~$\times$~16~mm slab of material was simulated for three different placements of the subgridding region: enclosing, traversing and away from the slab. The test was done for copper and a lossy dielectric with conductivity of 5~S/m and relative permittivity of 2. As a reference, uniformly discretized all-coarse and all-fine simulations were performed at the fine time step, in addition to the subgridding simulations. Coarse and fine meshes in uniformly discretized and subgridding runs were chosen as 1~mm and 0.2~mm respectively. 15~mm-thick perfectly matched layer (PML) terminated the simulation region. Modulated Gaussian magnetic current excitation was used at 15.0~GHz central frequency with 8.82~GHz half-width at half-maximum bandwidth. The time step was chosen as 0.467~ps in all test cases.

The magnetic field waveforms at the probe recorded in the different subgridding scenarios are shown in Fig.~\ref{fig:result_mt}, and are in excellent agreement among each other. This result confirms that the proposed subgridding method can properly handle material traverse, for both very good conductors and for lossy dielectrics.

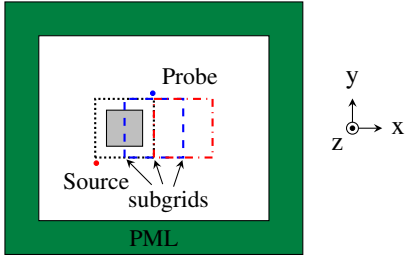
\begin{figure}[t]
	\centering
		\begin{tikzpicture}
		[arr/.style={->,>=stealth,shorten >=1pt}]
		
\pgfmathtruncatemacro{\PML}{15}
\pgfmathsetmacro{\lenx}{3.96}
\pgfmathsetmacro{\leny}{3.36}
\pgfmathsetmacro{\curx}{1.215}
\pgfmathsetmacro{\cury}{1.215}
\pgfmathsetmacro{\probex}{1.965}
\pgfmathsetmacro{\probey}{2.145}
\pgfmathsetmacro{\PMLax}{0.45}
\pgfmathsetmacro{\PMLay}{0.45}
\pgfmathsetmacro{\PMLbx}{3.51}
\pgfmathsetmacro{\PMLby}{2.91}
\pgfmathsetmacro{\encax}{1.2}
\pgfmathsetmacro{\encay}{1.29}
\pgfmathsetmacro{\encbx}{1.98}
\pgfmathsetmacro{\encby}{2.07}
\pgfmathsetmacro{\travax}{1.59}
\pgfmathsetmacro{\travay}{1.29}
\pgfmathsetmacro{\travbx}{2.37}
\pgfmathsetmacro{\travby}{2.07}
\pgfmathsetmacro{\outsax}{1.98}
\pgfmathsetmacro{\outsay}{1.29}
\pgfmathsetmacro{\outsbx}{2.76}
\pgfmathsetmacro{\outsby}{2.07}
\pgfmathsetmacro{\subgrlabx}{2.175}
\pgfmathsetmacro{\subgrlaby}{0.702}
\pgfmathsetmacro{\Rbig}{0.03}
\pgfmathsetmacro{\objax}{1.35}
\pgfmathsetmacro{\objay}{1.44}
\pgfmathsetmacro{\objbx}{1.83}
\pgfmathsetmacro{\objby}{1.92}

			\pgfmathsetmacro{\sizefont} {9}
			%PML	
			\draw [fill = {rgb:black,25;green,50;blue,25}] (0, 0) rectangle (\lenx,\leny);
			\draw [fill = white] (\PMLax, \PMLay) rectangle(\PMLbx, \PMLby);
			\draw[font=\fontsize{\sizefont}{\sizefont}\selectfont] (\lenx/2,\PMLay/2) node {PML};
			%Source
			\draw [red, fill = red] (\curx,\cury) circle (\Rbig);
			\draw[font=\fontsize{\sizefont}{\sizefont}\selectfont] (\curx,\cury) node [below] {Source};	
			%Probe
			\draw [blue, fill = blue] (\probex,\probey) circle (\Rbig);
			\draw[font=\fontsize{\sizefont}{\sizefont}\selectfont] (\probex,\probey) node [above right] {Probe};	
			
			%object
			\draw [fill = gray!50] (\objax, \objay) rectangle (\objbx, \objby);
			%subgrids
			\draw [densely dotted, black, thick] (\encax, \encay) rectangle(\encbx, \encby);
			\draw [dashed, blue, thick] (\travax, \travay) rectangle(\travbx, \travby);
			\draw [dashdotted, red, thick] (\outsax, \outsay) rectangle(\outsbx, \outsby);

			% subgrid
			
			\draw[arr, shorten <=9pt] (\encbx/2+\travbx/2,\subgrlaby) -- (\encax/2+\encbx/2,\outsay);
			\draw[arr, shorten <=7pt] (\encbx/2+\travbx/2,\subgrlaby) -- (\travax/2+\travbx/2,\outsay);
			\draw[arr, shorten <=6pt] (\encbx/2+\travbx/2,\subgrlaby) -- (\outsax/2+\outsbx/2,\outsay);
			\draw[font=\fontsize{\sizefont}{\sizefont}\selectfont] (\encbx/2+\travbx/2,\subgrlaby) node {\contour{white}{subgrids}};

			%axis
			\pgfmathsetmacro{\axisoffs}{\lenx/6}
			\draw[->,black,>=stealth] (\lenx + \axisoffs,\leny/2) -- (\lenx + \lenx/10+\axisoffs,\leny/2);
			\draw[->,black,>=stealth] (\lenx + \axisoffs,\leny/2) -- (\lenx + \axisoffs,\leny/2 + \lenx/10);
			\draw [fill=white] (\lenx + \axisoffs,\leny/2) circle (0.02*\lenx);
			\draw [fill=black] (\lenx + \axisoffs,\leny/2) circle (0.0075*\lenx);
			\draw[label] (\lenx + \axisoffs+\lenx/10,\leny/2) node [right] {x};
			\draw[label] (\lenx + \axisoffs,\leny/2 + \lenx/10) node [above] {y};
			\draw[label] (\lenx + \axisoffs,\leny/2) node [below left] {z};
		
		\end{tikzpicture}
	\caption{Layout used for the material traverse test of Sec.~\ref{sec:ne_mt}, and the three different placements of the subgridding region.}
	\label{fig:layout_mt}
\end{figure}

\begin{figure}[t]
	\centering
	\includegraphics[width=0.95\columnwidth]{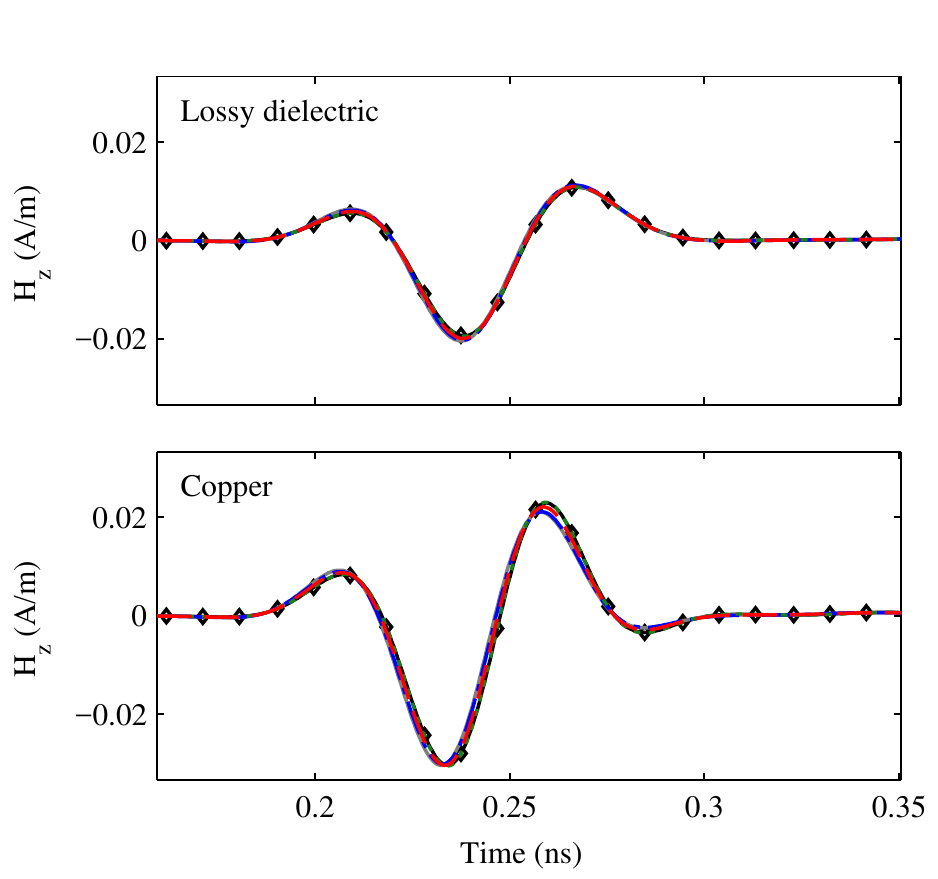}
	%\centering
%	\includegraphics[width=0.75\columnwidth]{jnl_scalar_mt_legend}
%	\caption{Time-domain probe waveform recorded for different locations of the subgrid with respect to the object. SG stands for subgridding}
	\caption{Time-domain magnetic field at the probe recorded for the three locations of the subgrid in Sec.~\ref{sec:ne_mt}: enclosed~(\color{legendBlue} \st{$\;$}~\st{$\;\;\;\;$}~\st{$\;$}\color{black}), outside~(\color{legendGreen} \st{$\;$}~\st{$\;$}~\st{$\;$} \st{$\;$}\color{black}) and traversing the object~(\color{legendRed} \st{$\;\;\;$}~\st{$\;\;\;$}\color{black}). Waveforms from the uniformly discretized all-coarse~(\st{$\;\;\;\diamond\;\;\;$}) and all-fine~(\color{legendGrey} \st{$\;\;\;\;\;\;\;$}\color{black}) simulations are also shown.}
	\label{fig:result_mt}
\end{figure}

\subsection{Application to Simulating Scatterer Reflections}
\label{sec:ne_rods4}

In order to investigate the accuracy of the proposed scheme, we looked at waveguide reflections from the scatterer shown in Fig.~\ref{fig:rods4}, which consisted of four copper rods with 1~mm radius. We have also investigated the reflections from the subgridding interface, in order to assess the quality of the subgridding scheme. The coarse cell was chosen to be 1~mm, which was $\frac{1}{10}$ of the minimum wavelength of interest that corresponded to 30.0~GHz. The fields of the incident wave were computed by running a reference simulation without the scatterer and without subgridding. Reflected wave fields were found by subtracting the incident wave fields from the total fields in simulations with the scatterer. 15~mm-thick PML boundary was chosen to terminate the two sides of the waveguide. The time steps in the subgridding runs were set 1\% below the CFL limit of the refined region. Uniformly discretized simulations were run 1\% below the CFL limits of the respective grids. 

The resulting reflections are shown in Fig.~\ref{fig:result_4rods} for the subgridding case with different refinement ratios, as well as the reference run with full refinement by a factor of 6. The simulation times are shown in Table~\ref{table:result_rods4}. It can be seen that the local refinement of the grid around the scatterer with the proposed subgridding method can improve the accuracy substantially compared to the coarse grid run. The larger choice of refinement ratio makes the solution very close to the reference all-fine solution. Moreover, very good speedup - almost by a factor of 11 - is achieved even for the grid refinement of 6 when the redundant high resolution of the grid in air is eliminated. The reflections from subgridding interface were significantly lower than those from the scatterer, further demonstrating the accuracy of the proposed method.
 
\begin{figure}[t]
	\centering
	\begin{tikzpicture}
	[label/.style={font=\fontsize{\sizefont}{\sizefont}\selectfont}]

\pgfmathsetmacro{\lenx}{5.28}
\pgfmathsetmacro{\leny}{3.2}
\pgfmathsetmacro{\PMLax}{1.2}
\pgfmathsetmacro{\PMLbx}{4.08}
\pgfmathsetmacro{\subgrax}{3.12}
\pgfmathsetmacro{\subgray}{1.28}
\pgfmathsetmacro{\subgrbx}{3.76}
\pgfmathsetmacro{\subgrby}{1.92}
\pgfmathsetmacro{\rodCoordAx}{3.28}
\pgfmathsetmacro{\rodCoordAy}{1.76}
\pgfmathsetmacro{\rodCoordBx}{3.28}
\pgfmathsetmacro{\rodCoordBy}{1.44}
\pgfmathsetmacro{\rodCoordCx}{3.6}
\pgfmathsetmacro{\rodCoordCy}{1.76}
\pgfmathsetmacro{\rodCoordDx}{3.6}
\pgfmathsetmacro{\rodCoordDy}{1.44}
\pgfmathsetmacro{\curx}{1.36}
\pgfmathsetmacro{\interfx}{1.52}
\pgfmathsetmacro{\rodR}{0.08}
\pgfmathtruncatemacro{\truelenxmm}{66}
\pgfmathtruncatemacro{\truelenymm}{40}
\pgfmathtruncatemacro{\dxCmm}{1}
\pgfmathtruncatemacro{\truedistInterfToSubgrmm}{20}
\pgfmathtruncatemacro{\truesubgrxmm}{8}
\pgfmathtruncatemacro{\truesubgrymm}{8}
\pgfmathtruncatemacro{\trueRodToRodNearestmm}{2}
\pgfmathtruncatemacro{\trueRodDiammm}{2}

	\pgfmathsetmacro{\sizefont} {9};
	
	%PML	
	\draw [fill = {rgb:black,25;green,50;blue,25}] (0, 0) rectangle (\lenx,\leny);
	\draw [fill = white] (\PMLax, 0) rectangle(\PMLbx, \leny);
	
	\draw[font=\fontsize{\sizefont}{\sizefont}\selectfont] (\PMLax/2,\leny/2) node {PML};
	\draw[font=\fontsize{\sizefont}{\sizefont}\selectfont] (\lenx - \PMLax/2,\leny/2) node {PML};
	
	%subgr			
	\draw[dashed] (\subgrax, \subgray) rectangle (\subgrbx, \subgrby);
	\draw[label] (\subgrax/2 + \subgrbx/2, \subgrby) node [above] {\truesubgrxmm$\times$\truesubgrymm\ mm};
	
	%rods
	\draw [fill = gray!50] (\rodCoordAx, \rodCoordAy) circle (\rodR);
	\draw [fill = gray!50] (\rodCoordBx, \rodCoordBy) circle (\rodR);
	\draw [fill = gray!50] (\rodCoordCx, \rodCoordCy) circle (\rodR);
	\draw [fill = gray!50] (\rodCoordDx, \rodCoordDy) circle (\rodR);
	
	%current source and interface
	\draw [thick, red] (\curx,0) -- (\curx, \leny);
	\draw [thick, gray!150] (\interfx,0) -- (\interfx, \leny);
	
	\draw [->,red,>=stealth, shorten >= 1pt] (\curx + 0.3, 0.85*\leny+0.2) -- (\curx, 0.85*\leny);
	\draw [label] (\curx + 0.3 -0.07, 0.85*\leny+0.2) node [right] {$J_y$ current};
	
	\draw [->,gray!150,>=stealth, shorten >= 1pt] (\interfx + 0.3, 0.73*\leny+0.2) -- (\interfx, 0.73*\leny);
	\draw [label] (\interfx + 0.3 -0.07, 0.73*\leny+0.2) node [right] {Probes};
	
	%PEC label
	
	\draw [very thick] (0,0) -- (\lenx, 0);
	\draw [very thick] (0,\leny) -- (\lenx, \leny);
	\draw[label] (\lenx/2 - 0.6, \leny + 0.4) node {PEC};
	\draw[->,black,>=stealth, shorten >= 3pt, shorten <= 8pt] (\lenx/2 - 0.6, \leny + 0.4) -- (\lenx/2, \leny);
	
			%axis
			\pgfmathsetmacro{\axisoffs}{\lenx/6}
			\draw[->,black,>=stealth] (\lenx + \axisoffs,\leny*0.7) -- (\lenx + \lenx/10+\axisoffs,\leny*0.7);
			\draw[->,black,>=stealth] (\lenx + \axisoffs,\leny*0.7) -- (\lenx + \axisoffs,\leny*0.7 + \lenx/10);
			\draw [fill=white] (\lenx + \axisoffs,\leny*0.7) circle (0.02*\lenx);
			\draw [fill=black] (\lenx + \axisoffs,\leny*0.7) circle (0.0075*\lenx);
			\draw[label] (\lenx + \axisoffs+\lenx/10,\leny*0.7) node [right] {x};
			\draw[label] (\lenx + \axisoffs,\leny*0.7 + \lenx/10) node [above] {y};
			\draw[label] (\lenx + \axisoffs,\leny*0.7) node [below left] {z};
	
	%label dimensions
	\pgfmathsetmacro{\d}{\leny/30};
	\draw[<->,black,>=stealth] (0,-\d) -- (\lenx, -\d);	
	\draw[label] (\lenx/2, -\d) node [below] {\truelenxmm\ mm};

	\draw[<->,black,>=stealth] (-\d,0) -- (-\d, \leny);	
	\draw[label] (-\d, \leny/2) node [rotate = 90, above] {\truelenymm\ mm};
	
	% label mesh
	%\draw[label] (\interfx,0) node [above right] {$\dx=\dy=$\dxCmm mm};
	\draw[label] (\interfx,0) node [above right] {\begin{tabular}{l} $\dx$ = \dxCmm\ mm \\ $\dy$ = \dxCmm\ mm \end{tabular}};
	
	% label distance to subgr
	\draw[<->,black,>=stealth, shorten >= 0pt, shorten <= 0pt] (\interfx,\leny/2) -- (\subgrax, \leny/2);	
	\draw[label] (\interfx/2 + \subgrax/2, \leny/2) node [below] {\truedistInterfToSubgrmm\ mm};
	
	% draw the subgridding region separately
	\pgfmathsetmacro{\offsx}{\lenx+\lenx/15};
	\pgfmathsetmacro{\offsy}{0};
	
	\pgfmathsetmacro{\subexpand}{2};
	\pgfmathsetmacro{\subgrlenx}{(\subgrbx - \subgrax)};
	\pgfmathsetmacro{\subgrleny}{(\subgrby - \subgray)};
	\pgfmathsetmacro{\subgrrodAx}{(\subexpand*\rodCoordAx - \subexpand*\subgrax + \offsx)};
	\pgfmathsetmacro{\subgrrodAy}{(\subexpand*\rodCoordAy - \subexpand*\subgray + \offsy)};
	\pgfmathsetmacro{\subgrrodBx}{(\subexpand*\rodCoordBx - \subexpand*\subgrax + \offsx)};
	\pgfmathsetmacro{\subgrrodBy}{(\subexpand*\rodCoordBy - \subexpand*\subgray + \offsy)};
	\pgfmathsetmacro{\subgrrodCx}{(\subexpand*\rodCoordCx - \subexpand*\subgrax + \offsx)};
	\pgfmathsetmacro{\subgrrodCy}{(\subexpand*\rodCoordCy - \subexpand*\subgray + \offsy)};
	\pgfmathsetmacro{\subgrrodDx}{(\subexpand*\rodCoordDx - \subexpand*\subgrax + \offsx)};
	\pgfmathsetmacro{\subgrrodDy}{(\subexpand*\rodCoordDy - \subexpand*\subgray + \offsy)};
%	\draw[dashed] (\offsx, \offsy) rectangle (\offsx+\subexpand*\subgrlenx, \offsy+\subexpand*\subgrleny);
	\draw[label] (\offsx + \subexpand*\subgrlenx/2, \offsy) node [below] {Scatterer};
	
	%draw the rods
\draw [fill = gray!50] (\subgrrodAx, \subgrrodAy) circle (\subexpand*\rodR);
\draw [fill = gray!50] (\subgrrodBx, \subgrrodBy) circle (\subexpand*\rodR);
\draw [fill = gray!50] (\subgrrodCx, \subgrrodCy) circle (\subexpand*\rodR);
\draw [fill = gray!50] (\subgrrodDx, \subgrrodDy) circle (\subexpand*\rodR);	
	%label distance between rods
\draw [<->, >=stealth, shorten >= 2pt, shorten <= 2pt] (\subgrrodCx, \subgrrodCy-\rodR) -- (\subgrrodDx, \subgrrodDy+\rodR);
\draw [label] (\subgrrodCx, \subgrrodCy/2+\subgrrodDy/2) node [right] {\trueRodToRodNearestmm\ mm};
	%label radius of the rods
	\pgfmathsetmacro{\cos}{0.7071};
	\pgfmathsetmacro{\sin}{0.7071};
	 \draw [<->, >=stealth] (\subgrrodAx- \cos*\rodR*\subexpand, \subgrrodAy - \sin*\rodR*\subexpand) -- (\subgrrodAx+ \cos*\rodR*\subexpand, \subgrrodAy + \sin*\rodR*\subexpand);
	 
	 \draw [label] (\subgrrodAx+\cos*\rodR*\subexpand, \subgrrodAy + \sin*\rodR*\subexpand) node [above] {\trueRodDiammm\ mm}; 
	
	\end{tikzpicture}
	\caption{Layout of the four-rod reflection simulation discussed in Sec.~\ref{sec:ne_rods4}. The dashed line shows the location of the subgrid in the subgridding run.}
	\label{fig:rods4}
\end{figure}
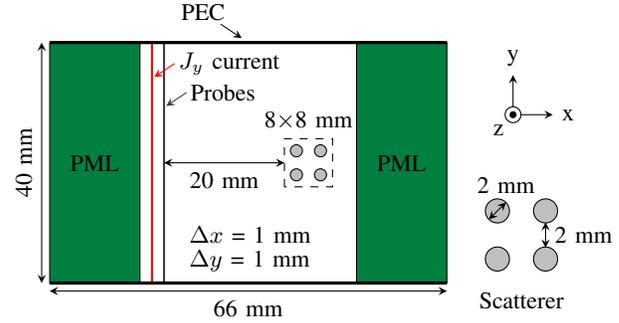

\begin{figure}[t]
	\centering
	\includegraphics[width=0.95\columnwidth]{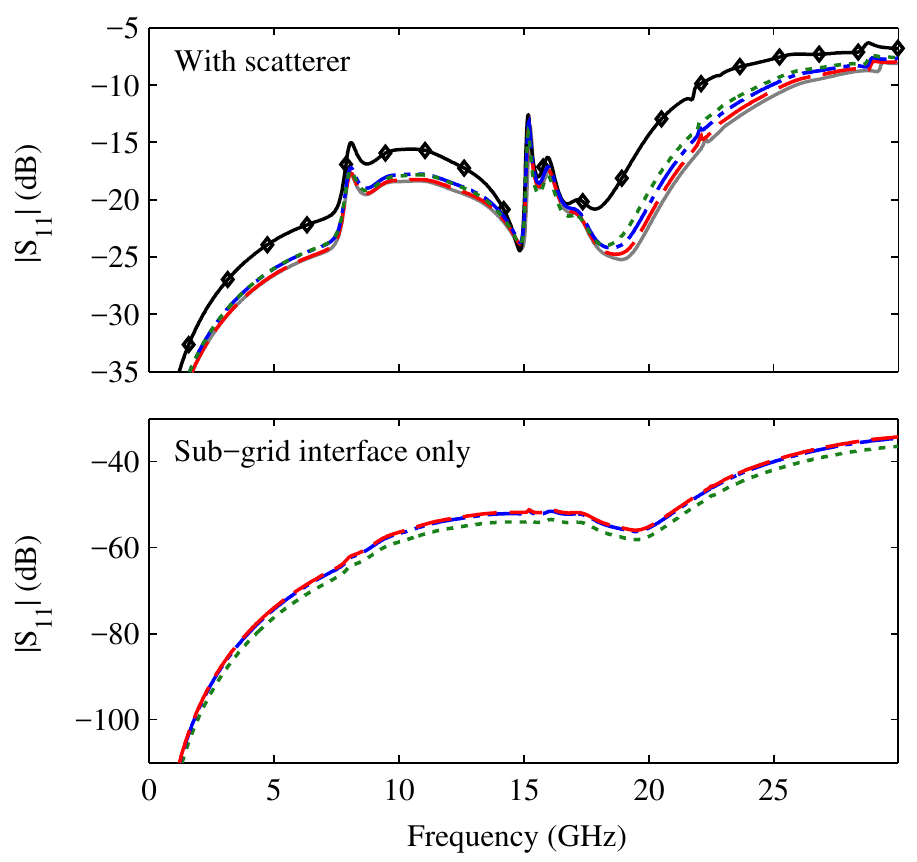}
	%\centering
	%\includegraphics[width=0.75\columnwidth]{jnl_scalar_4rods_legend}
	\caption{Reflected power with respect to the incident for the example of Sec.~\ref{sec:ne_rods4}. Top panel: reflections from the four-rod scatterer for different global discretization: all-coarse~(\st{$\;\;\;\diamond\;\;\;$}) and all-fine with $r$~=~6~(\color{legendGrey} \st{$\;\;\;\;\;\;\;$}\color{black}); and for the subgridding runs: $r=$~2~(\color{legendGreen} \st{$\;$}~\st{$\;$}~\st{$\;$}~\st{$\;$}\color{black}), $r$~=~4~(\color{legendBlue} \st{$\;$}~\st{$\;\;\;\;$}~\st{$\;$}\color{black}) and $r$~=~6~(\color{legendRed}\st{$\;\;\;$}~\st{$\;\;\;$}\color{black}). Bottom panel: reflections from the subgrid interface only.}
	\label{fig:result_4rods}
\end{figure}

\begin{table}[t]
	\centering
	\caption{Simulation times for different meshes in four-rod scatterer simulations in Sec.~\ref{sec:ne_rods4}. Simulation times were recorded for the total field run only - not for the reference run.}
	\begin{tabular}{c c c c} 
		\hline
		Method		&	Simulation time (s) \\ 
		\hline\hline
		All-fine ($r$ = 6)	&	427.4	 \\ 
		\hline
		All-coarse	&	3.2 \\ 
		\hline
		Subgridding ($r$ = 2)	&	10.7	 \\ 
		\hline
		Subgridding	($r$ = 4)&	23.3	 \\ 
		\hline
		Subgridding	($r$ = 6)&	39.5	 \\ 
		\hline
	\end{tabular}
	\label{table:result_rods4}
\end{table}
\subsection{Application to Exposure Studies}
\label{sec:ne_hm}
We show the possibility of applying the method for multiscale human exposure simulations. The chosen setup is shown in Fig.~\ref{fig:layout_hm}. A transverse cross-section of the head of the model of Ella from IT'IS Virtual Population V1.x~\cite{VirtualPopulation} was used to assign the material properties~\cite{material_properties} to the FDTD cells. A 900~MHz source was placed approximately 3 meters away from the human head. The simulation region was terminated with 20~cm PML. The reference (fine) resolution in the tissues was set to 2~mm, based on the mesh size chosen by~\cite{FinalReport2009} in a radiation exposure study at that frequency. Specific absorption rate, or SAR, was evaluated according to the formula in~\cite{SARref}, which was used as follows for sinusoidal excitation
\begin{equation}
\centering
SAR = \frac{\sigma (E_{x_{p}}^2 + E_{y_{p}}^2)}{2\rho}\,,
\end{equation}
where subscript "$p$" denotes the peak absolute value of a field component and $\rho$ corresponds to tissue density. SAR was calculated for each of the primary cells. The values of the electric field components at the primary cell centers were found by averaging the nearest known samples at the cell edges. The peak values of the fields were found for the time interval from 25.6~ns to 26.8~ns, which gave the wave sufficient time to reach the head and penetrate inside it.

In the subgridding run, the empty space was coarsened to 1~cm, which corresponded to $\frac{1}{33.3}$ of the wavelength. The reference all-fine run was performed at 2~mm resolution, along with the all-coarse run where the entire simulation region was discretized at 1~cm. The volumetric integral of SAR over the tissues was used as an accuracy metric
\begin{equation}
\sum_{i} \sum_{j} {SAR|_{i+\half, j+\half}} \dx \dy \,,
\end{equation}
where $\dx$ and $\dy$ are FDTD cell dimensions in the tissues. Time step of 4.67~ps was chosen for the all-fine simulation and for the subgridding simulation. The coarse grid case was run at 23.11~ps.

\begin{figure}[t]
	\centering
	\begin{tikzpicture}
		[label/.style={font=\fontsize{\sizefont}{\sizefont}\selectfont}]
	
\pgfmathsetmacro{\Scale}{0.0225}
\pgfmathsetmacro{\lenx}{3.609}
\pgfmathsetmacro{\leny}{2.745}
\pgfmathsetmacro{\curx}{0.3645}
\pgfmathsetmacro{\cury}{1.3455}
\pgfmathsetmacro{\PMLax}{0.18}
\pgfmathsetmacro{\PMLay}{0.18}
\pgfmathsetmacro{\PMLbx}{3.429}
\pgfmathsetmacro{\PMLby}{2.565}
\pgfmathsetmacro{\subgrax}{3.042}
\pgfmathsetmacro{\subgray}{1.26}
\pgfmathsetmacro{\subgrbx}{3.249}
\pgfmathsetmacro{\subgrby}{1.485}
\pgfmathsetmacro{\threeMeters}{2.7}
\pgfmathsetmacro{\truelenx}{4.01}
\pgfmathsetmacro{\trueleny}{3.05}
	
				\pgfmathsetmacro{\Rbig}{0.035};
				\pgfmathsetmacro{\sizefont} {9}
				
				%PML	
				\draw [fill = {rgb:black,25;green,50;blue,25}] (0, 0) rectangle (\lenx,\leny);
				\draw [fill = white] (\PMLax, \PMLay) rectangle(\PMLbx, \PMLby);
				\draw[label] (\lenx/2-3*\PMLay,2.5*\PMLay) node  {PML};
				\draw[->,>=stealth, shorten >= 1pt, shorten <= 9pt] (\lenx/2-3*\PMLay,2.5*\PMLay) -- (\lenx/2,\PMLay/2);
					
				%Source
				\draw [red, fill = red] (\curx,\cury) circle (\Rbig);
				\draw[label] (\curx,\cury) node [right] {Source};	
				
				%subgridding
				\node[inner sep=0pt] (whitehead) at (\subgrax/2+\subgrbx/2,\subgray/2+\subgrby/2)
				{\includegraphics[scale=\Scale]{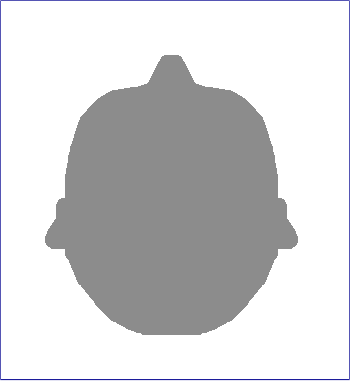}};
				\draw [dashed] (\subgrax, \subgray) rectangle (\subgrbx,\subgrby);
				\draw [label] (\PMLbx, \subgrby) node [above left] {subgrid};
				\draw [label] (\subgrax, \subgray/2+\subgrby/2) node [below left] {\begin{tabular}{c}Head\\slice\end{tabular}};
				\draw [->,>=stealth, shorten >= 1pt] (\subgrax - \lenx/15, \subgray/2+\subgrby/2- \lenx/15) -- (\subgrax/2+\subgrbx/2, \subgray/2+\subgrby/2);
		
			%	%label distance between the head and the source
			%	\draw[<->, >=stealth] (\curx, \cury) -- (\curx + \threeMeters, \cury);
			%	\draw[label] (\curx+0.5*\threeMeters, \cury) node [below] {$3$ m};				
	
				%axis
				\pgfmathsetmacro{\axisoffs}{\lenx/6}
				\draw[->,black,>=stealth] (\lenx + \axisoffs,\leny/2) -- (\lenx + \lenx/10+\axisoffs,\leny/2);
				\draw[->,black,>=stealth] (\lenx + \axisoffs,\leny/2) -- (\lenx + \axisoffs,\leny/2 + \lenx/10);
				\draw [fill=white] (\lenx + \axisoffs,\leny/2) circle (0.02*\lenx);
				\draw [fill=black] (\lenx + \axisoffs,\leny/2) circle (0.0075*\lenx);
				\draw[label] (\lenx + \axisoffs+\lenx/10,\leny/2) node [right] {x};
				\draw[label] (\lenx + \axisoffs,\leny/2 + \lenx/10) node [above] {y};
				\draw[label] (\lenx + \axisoffs,\leny/2) node [below left] {z};
	
	%label dimensions
	\pgfmathsetmacro{\d}{\leny/30};
	\draw[<->,black,>=stealth] (0,-\d) -- (\lenx, -\d);	
	\draw[label] (\lenx/2, -\d) node [below] {\truelenx m};
	
	\draw[<->,black,>=stealth] (-\d,0) -- (-\d, \leny);	
	\draw[label] (-\d, \leny/2) node [rotate = 90, above] {\trueleny m};
	
	\end{tikzpicture}
	\caption{Layout of the simulation in Sec.~\ref{sec:ne_hm} with human head cross-section placed approximately 3 meters away from a point source.}
	\label{fig:layout_hm}
	
\end{figure}
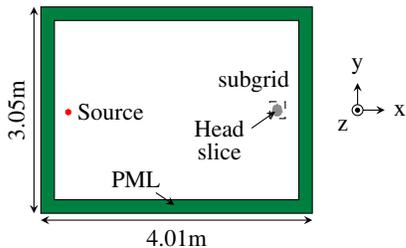

\begin{table}[t]
	\centering
	\caption{Error in the integral of SAR and simulation times in the simulations discussed in Sec.~\ref{sec:ne_hm}.}
	\begin{tabular}{c c c c} 
		\hline
		Method		&	Error in SAR integral		&	Simulation time (s) \\ 
		\hline\hline
		All-fine	&			Not applicable		&	1596.4	 \\ 
		\hline
		All-coarse	&				59.5\%		&	6.9	 \\ 
		\hline
		Subgridding	&			-3.1\%			&	34.5	 \\ 
		\hline
	\end{tabular}
	\label{table:result_hm}
\end{table}

\begin{figure}[t]

	\centering
	\begin{tikzpicture}
	\pgfmathsetmacro{\fraccolwid}{0.27};
	\pgfmathsetmacro{\dist}{2.8};
	
	\node[inner sep=0pt] (whitehead) at (0,0) [below] {\includegraphics[width = \fraccolwid\columnwidth]{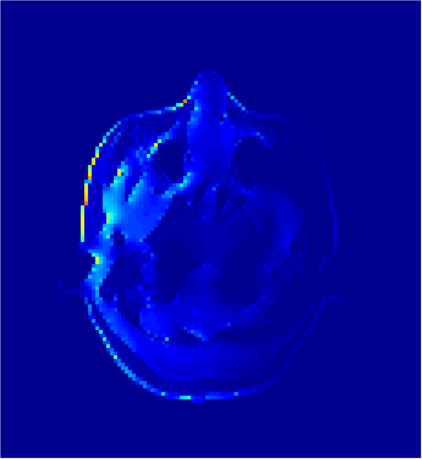}};
	\draw[label] (0,0) node [above] {All-fine};
	
	\node[inner sep=0pt] (whitehead) at (\dist,0)[below] {\includegraphics[width = \fraccolwid\columnwidth]{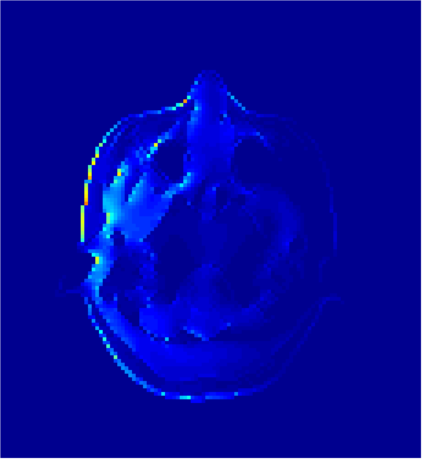}};
	\draw[label] (\dist,0) node [above] {Subgridding};
	
	\node[inner sep=0pt] (whitehead) at (2*\dist,0)[below] {\includegraphics[width = \fraccolwid\columnwidth]{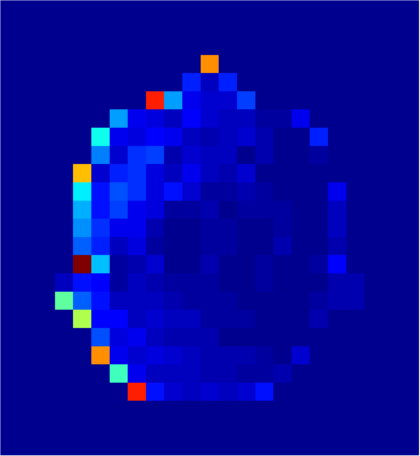}};
	\draw[label] (2*\dist,0) node [above] {All-coarse};
	\end{tikzpicture}
	
	\caption{SAR maps obtained in Sec.~\ref{sec:ne_hm} with all-fine (2~mm), all-coarse (1~cm) discretization and with subgridding (2~mm in the tissues and 1~cm in the air).}
	\label{fig:result_hm}
\end{figure}
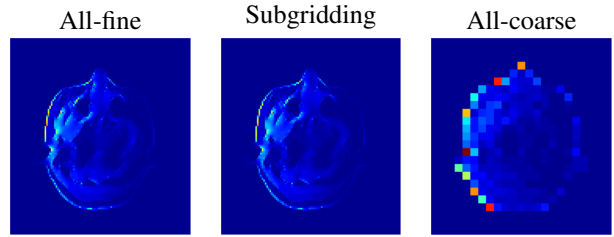

The resulting SAR maps are shown in Fig.~\ref{fig:result_hm}. Table~\ref{table:result_hm} shows simulation times and percent error in the total SAR with respect to the all-fine simulation. The results show that with subgridding a speedup of 46X can be obtained with only 3.1\% loss in accuracy. No noticeable difference can be seen on the SAR maps in Fig.~\ref{fig:result_hm} between the subgridding and all-fine simulations. When, instead, the coarse resolution was chosen for the entire grid, the head tissues were not sufficiently resolved and the integral SAR differed from the reference by 59.5\%, showing the necessity of local grid refinement. 

\section{Conclusion}
This paper proposed a dissipative systems theory for FDTD, recognizing that FDTD equations can be seen as a dynamical system which is dissipative under a generalized Courant-Friedrichs-Lewy condition. The theory provides a new, powerful framework to make FDTD stability analysis simpler and modular. Stability conditions can indeed be given on the individual components (e.g. boundary conditions, meshes, thin-wire models) rather than on the whole coupled FDTD setup. The theory is intuitive since rooted on the familiar concept of energy dissipation, and sheds new light on the root mechanisms behind FDTD instability.  As an example of application, a simple yet effective subgridding algorithm is derived, with straightforward stability proof. The proposed algorithm allows material traverse, is simpler to implement than existing solutions, and supports an arbitrary grid refinement ratio. Numerical results confirm its stability and accuracy. Speedups of up to 46X were observed with only 3.1\% error with respect to standard FDTD.
%%%%%%%%%%%%%%%%%
\bibliographystyle{myIEEEtran}
\bibliography{IEEEabrv,bibliography}

\end{document}

%% file: myDefinitions.tex
\newcommand{\vect}[1]{\mathbf{#1}}

\newcommand{\matr}[1]{\mathbf{#1}}

\newcommand{\junk}[1] {}

\def\XXint#1#2#3{{\setbox0=\hbox{$#1{#2#3}{\int}$}
\vcenter{\hbox{$#2#3$}}\kern-.5\wd0}}

\newcommand*\widebar[1]{%
  \hbox{%
    \vbox{%
      \hrule height 0.5pt % The actual bar
      \kern0.3ex%         % Distance between bar and symbol
      \hbox{%
        \kern-0.05em%      % Shortening on the left side
        \ensuremath{#1}%
        \kern-0.05em%      % Shortening on the right side
      }%
    }%
  }%
}